\documentclass[reqno,12pt]{amsart}

\usepackage{amsmath}
\usepackage{latexsym}
\usepackage{amssymb}
\usepackage{hyperref}
\usepackage{graphicx}
\usepackage{epstopdf}
\usepackage{ifpdf}
\ifpdf
\fi

\makeatletter
 
 \newcommand{\Rmnum}[1]{\expandafter\@slowromancap\romannumeral #1@}
 \makeatother

\newtheorem{theorem}{Theorem}[section]

\newtheorem{proposition}[theorem]{Proposition}
\newtheorem{lemma}[theorem]{Lemma}
\newtheorem{corollary}[theorem]{Corollary}
\newtheorem{remark}[theorem]{Remark}

\newcommand{\R}{{\mathbb R}}

\newcommand{\Z}{{\mathbb Z}}
\newcommand{\C}{{\mathbb C}}

\newcommand{\be}{\begin{equation}}
\newcommand{\ee}{\end{equation}}
\newcommand{\bea}{\begin{eqnarray}}
\newcommand{\eea}{\end{eqnarray}}
\newcommand{\ba}{\begin{array}}
\newcommand{\ea}{\end{array}}

\newcommand{\ol}{\overline}

\newcommand{\id}{\mathbb{I}}

\newcommand{\re}{\mathrm{Re}}
\newcommand{\im}{\mathrm{Im}}

\newcommand{\eps}{\varepsilon}

\newcommand{\sig}{\sigma}
\newcommand{\Sig}{\Sigma}
\newcommand{\lam}{\lambda}

\newcommand{\gam}{\gamma}
\newcommand{\Gam}{\Gamma}

\newcommand{\x}{\xi}
\newcommand{\dta}{\delta}
\newcommand{\Dta}{\Delta}

\newcommand{\tha}{\theta}

\newcommand{\varplon}{\varepsilon}

\numberwithin{equation}{section}

\begin{document}
\title[Long-time asymptotics for DNLS equation]{Long-time asymptotic for the derivative nonlinear Schr\"odinger equation with decaying initial value}

\author[J.Xu]{Jian Xu}
\address{School of Mathematical Sciences\\
Fudan University\\
Shanghai 200433\\
People's  Republic of China}
\email{11110180024@fudan.edu.cn}

\author[E.Fan]{Engui Fan*}
\address{School of Mathematical Sciences, Institute of Mathematics and Key Laboratory of Mathematics for Nonlinear Science\\
Fudan University\\
Shanghai 200433\\
People's  Republic of China}
\email{correspondence author:faneg@fudan.edu.cn}


\date{\today}

\begin{abstract}
We present a new Riemann-Hilbert problem formalism for the initial value problem for the derivative nonlinear Schr\"odinger (DNLS) equation:
 \[
 iq_t(x,t)+q_{xx}(x,t)+i(|q|^2q)_x=0
 \]
on the line. We show that the solution of this initial value problem can be obtained from the solution of some associated Riemann-Hilbert problem. This new Riemann-Hilbert problem for the DNLS equation will lead us to use nonlinear steepest-descent/stationary phase method or Deift-Zhou method to derive the long-time asymptotic for the DNLS equation on the line.

\end{abstract}

\maketitle

\section{Introduction}

The main purpose of this paper is to develop an inverse scattering approach, based on an appropriate Riemann-Hilbert problem formulation, for the initial value problem for the derivative nonlinear Schr\"odinger (DNLS) equation \cite{kn} on the line, whose form is:
\begin{subequations}\label{DNLSandIv}
\be \label{DNLS}
iq_t(x,t)+q_{xx}(x,t)-i(rq^2)_x=0,
\ee
\be \label{DNLSIv}
q(x,0)=q_0(x).
\ee
\end{subequations}
where $r=\pm {\bar q},\bar q $ denotes complex conjugate of $q$, the subscripts denote differentiation with respect to the corresponding variables. And in this paper we use $r=-\bar q$
\par
The DNLS equation has several applications in plasma physics. In plasma physcis, it is a model for Alfv$\acute e$n waves propagating parallel to the ambient magnetic field, $q$ being the transverse magnetic field perturbation and $x$ and $t$ being space and time coordinates, respectively \cite{m}.
\par
Our goal is to develop the inverse scattering approach to the DNLS equation, in view of its further application for studying the long-time asymptotics. The starting point of the approach is the Lax pair representation: the DNLS equation is indeed the compatibility condition of two linear equations \cite{kn}:
\begin{subequations} \label{DNLS-Lax}
\be \label{Lax-x}
v_{1x}+ik^2v_1=qk v_2,\quad v_{2x}-ik^2v_2=rk v_1,
\ee
\be \label{Lax-t}
iv_{1t}=Av_1+Bv_2,\quad iv_{2t}=Cv_1-Av_2,
\ee
\end{subequations}
where $\x \in \C$ is the spectral parameter, and
\be
\ba{lc}
A=2k^4+k^2rq,&B=2ik^3q-k q_x+ik rq^2, \\
C=2ik^3r-k r_x+ik r^2q.&
\ea
\ee
\par
In the present paper, we propose a scattering inverse scattering formalism, in which the Lax pair is used in the form of a system of first order matrix-valued linear equations. Then dedicated solutions of this system are defined and used to construct a Riemann-Hilbert (RH) problem in the complex plane. The main advantage of the representation of a solution of the DNLS equation in terms of the solution of a Riemann-Hilbert problem is that it allows applying the nonlinear steepest descent method by Deift and Zhou \cite{dz} in order to obtain rigorous results on the long-time asymptotic behavior of the solution.
\par
An alternative inverse scattering method based on a Riemann-Hilbert problem formulation for the DNLS equation can be founded in \cite{avkahv} for the Cauchy problem, and in \cite{l} for the initial-boundary value problem on the half-line, in \cite{xf} for the initial-boundary value problem on the interval.
\par
In Section 2, we define appropriate eigenfunctions and spectral functions, which are used in Section 3 in the reformulation of the scattering problem as a Riemann-Hilbert problem of analytic conjugation in the complex plane of the spectral parameter. And in Section 4 we obtain the long-time asymptotic behavior of the solution of DNLS equation by the method of the Deift-Zhou/nonlinear steepest descent based on the new Riemann-Hilbert problem which is obtained in the subsection \ref{newRHPsec}.

\section{Eigenfunctions and spectral functions}

\subsection{Eigenfunctions}

First introducing
\[
\ba{lcr}
\psi =\left(\ba{c}v_1\\v_2\ea \right),& Q =\left(\ba{lr}0&q\\r&0\ea \right),& \sig_3=\left(\ba{lr}1&0\\0&-1\ea \right), \ea
\]
we can rewrite the Lax pair (\ref{DNLS-Lax}) in a matrix form:
\be \label{Lax:DNLS2}
\ba{l}
\psi_x+ik^2 \sig_3 \psi=k Q \psi,   \\
\psi_t+2ik^4 \sig_3 \psi=(-ik^2 Q^2 \sig_3+2 k^3 Q-ik Q_x \sig_3+k Q^3) \psi,
\ea
\ee

\par
Extending the column vector $\psi$ to a $2\times 2$ matrix and letting
\[
\Psi=\psi e^{i(k^2 x+2k^4 t)\sig_3},
\]
we obtain the equivalent Lax pair
\bea \label{Lax:DNLS3}
  && \Psi_x+ik^2 [\sig_3,\Psi]=k Q \Psi,   \nonumber \\
  && \Psi_t+2ik^4 [\sig_3,\Psi]=(-ik^2 Q^2 \sig_3+2 k^3 Q-ik Q_x \sig_3+k Q^3) \Psi,
\eea
which can be written in full derivative form
\be \label{Lax:DNLS4}
d(e^{i(k^2 x+2k^4 t)\hat \sig_3}\Psi(x,t,k))=e^{i(k^2 x+2k^4 t)\hat \sig_3}U(x,t,k)\Psi,
\ee
where
\be
U=U_1dx+U_2dt=k Qdx+(-ik^2 Q^2 \sig_3+2 k^3 Q-ik Q_x \sig_3+k Q^3)dt.
\ee
\par
In order to formulate a Riemann-Hilbert problem for the solution of the inverse spectral problem,we seek solutions of the spectral problem which approach the $2\times 2$ identity matrix  as $k \rightarrow \infty$. It turns out that solutions of Eq.(\ref{Lax:DNLS4}) do not exhibit this property \cite{kn}, hence we have to transform the solution $\Psi$ of Eq.(\ref{Lax:DNLS4}) into the desired asympototic behavior \cite{l}.

Consider a solution of Eq.(\ref{Lax:DNLS4}) of the form
\[
\Psi=D+\frac{\Psi_1}{k}+\frac{\Psi_2}{k^2}+\frac{\Psi_3}{k^3}+O(\frac{1}{k^4}),\quad k \rightarrow \infty
\]
where $D,\Psi_1,\Psi_2,\Psi_3$ are independent of $k$. Substituting the above expansion into the the first equation of
(\ref{Lax:DNLS3}),and comparing the same order of $k$'s frequency, it follows from the $O(k^2)$ terms that D is a diagonal matrix. Furthermore,one finds the following equations for
 the $O(k)$ and the diagonal part of the $O(1)$ terms
\[
O(k):i[\sig_3,\Psi_1]=QD,\quad
i.e. \quad \Psi_1^{(o)}=\frac{i}{2} QD\sig_3,
\]
with $\Psi_1^{(o)}$ being the off-diagonal part of $\Psi_1$,and
\[
O(1):D_x=Q\Psi_1^{(o)},
\]
i.e.
\be \label{2.1.1}
D_x=\frac{i}{2}Q^2\sig_3D.
\ee
\par
On the other hand, substituting the above expansion into the second equation of (\ref{Lax:DNLS3}), one obtains from that
\be \label{2.1.2}
O(k^3):2i[\sig_3,\Psi_1]=2QD,\quad i.e. \quad \Psi_1^{(o)}=\frac{i}{2}QD\sig_3;
\ee
and
\be \label{2.1.3}
O(k): 2i[\sig_3,\Psi_3]=-iQ^2 \sig_3 \Psi_1^{(o)}+2Q\Psi_2^{(d)}-iQ_x \sig_3 D+Q^3D,
\ee
i.e.
\be \label{2.1.4}
-iQ^2 \sig_3 \Psi_2^{(d)}+2Q\Psi_3^{(o)}=-\frac{1}{2}Q^3\Psi_1^{(o)}+\frac{1}{2}QQ_xD+\frac{i}{2}Q^4\sig_3D,
\ee
where $\Psi_2^{(d)}$ denotes the diagonal part of $\Psi_2$; and for the diagonal part of the $O(1)$ terms
\[
O(1):D_t=-iQ^2 \sig_3 \Psi_2^{(d)}+2Q\Psi_3^{(o)}-iQ_x\sig_3\Psi_1^{(o)}+Q^3\Psi_1^{(o)},
\]
again,using (\ref{2.1.2}) and (\ref{2.1.4}),we have
\[
D_t=(\frac{3i}{4}Q^4\sig_3+\frac{1}{2}[Q,Q_x])D,
\]
which can be written in terms of $q$ and $r$ as
\be \label{2.1.5}
D_t=(\frac{3i}{4}r^2 q^2+\frac{1}{2}(r_xq-rq_x))\sig_3D.
\ee

We note that Eq.(\ref{DNLS}) admits the conservation law
\[
(\frac{i}{2}rq)_t=(\frac{3i}{4}r^2 q^2+\frac{1}{2}(r_xq-rq_x))_x.
\]
\par
Because we just consider the Cauchy problem for the DNLS equation (\ref{DNLS}), the two Eqs.(\ref{2.1.1}) and (\ref{2.1.5}) for $D$ are consistent and are both satisfied if we define
\be \label{D}
D(x,t)=e^{i\int_{(-\infty,t)}^{(x,t)}\Dta \sig_3},
\ee
where $\Dta$ is the closed real-valued one-form
\be \label{Dta}
\Dta(x,t)=\frac{1}{2}qrdx+(\frac{3}{4}r^2 q^2-\frac{i}{2}(r_xq-rq_x))dt .
\ee
Noting that the integral in (\ref{D}) is independent of the path of integration and the $\Dta$ is independent of $k$,
then we introduce a new function $\mu$ by
\be \label{mu}
\Psi(x,t,k)=e^{i\int_{(+\infty,t)}^{(x,t)}\Dta \hat \sig_3}\mu(x,t,k)D(x,t),
\ee
Thus,we have
\be \label{muinfty}
\mu=\id+O(\frac{1}{k}),\quad k \rightarrow \infty,
\ee
and the Lax pair of Eq.(\ref{Lax:DNLS4}) becomes
\be \label{Lax:DNLS5}
d(e^{i(k^2 x+2k^4 t)\hat \sig_3}\mu(x,t,k))=W(x,t,k),
\ee
where
\[
W(x,t,k)=e^{i(k^2 x+2k^4 t)\hat \sig_3}V(x,t,k)\mu,
\]
with
\[
V=V_1dx+V_2dt=e^{-i\int_{(+\infty,t)}^{(x,t)}\Dta \hat \sig_3}(U-i\Dta \sig_3).
\]
Taking into account the definition of $U$ and $\Dta$,we find that
\be \label{V1}
V_1=\left (\ba{ll} -\frac{i}{2}rq& k qe^{-2i\int_{(+\infty,t)}^{(x,t)}\Dta}\\
                  k re^{2i\int_{(+\infty,t)}^{(x,t)}\Dta}&\frac{i}{2}rq \ea
                  \right ),
\ee
\be \label{V2}
V_2=\left(\ba{ll} -ik^2rq-\frac{3i}{4}r^2 q^2-\frac{1}{2}(r_xq-rq_x)&(2k^3q+ik q_x+k q^2r)e^{-2i\int_{(+\infty,t)}^{(x,t)}\Dta}\\
                  (2k^3r+ik r_x+k r^2q)e^{2i\int_{(+\infty,t)}^{(x,t)}\Dta}&ik^2rq+\frac{3i}{4}r^2 q^2+\frac{1}{2}(r_xq-rq_x)
    \ea \right).
\ee
Then Eq.(\ref{Lax:DNLS5}) for $\mu$ can be written as
\bea \label{Lax:DNLS6}
&&\mu_x+ik^2[\sig_3,\mu]=V_1\mu,\nonumber \\
&&\mu_t+2ik^4[\sig_3,\mu]=V_2\mu.
\eea

Throughout this section we assume that $q(x,t)$ is sufficiently smooth,we define two solutions of Eq.(\ref{Lax:DNLS5}) by
\be \label{muj}
\mu_j(x,t,k)=\id+\int_{(x_j,t_j)}^{(x,t)}e^{-i(k^2 x+2k^4 t)\hat \sig_3}W(y,\tau,k),\qquad j=1,2,
\ee
where $(x_1,t_1)=(+\infty,t),(x_2,t_2)=(-\infty,t)$,see Figure 1.
\begin{figure}[th]
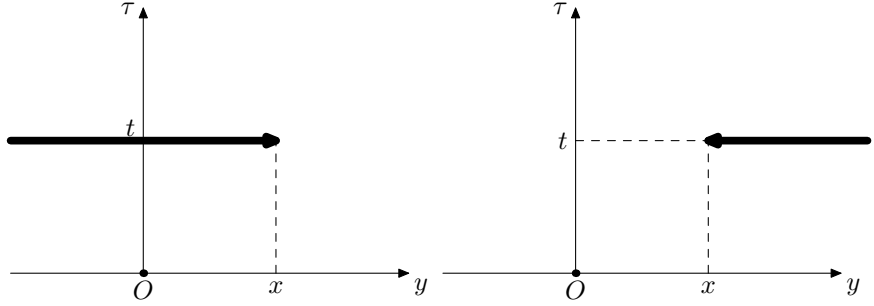

\centering
\includegraphics{IST.1}
\includegraphics{IST.2}
\caption{Paths integrals of the $\mu_1$ and $\mu_2$.}
\end{figure}

\par
The analytic properties of the $2\times 2$ matrices $\mu_j(x,t;k)$, $j=1,2$, that follow from (\ref{muj}), are collected in the following proposition. We denote by $\mu_j^{(1)}(x,t,k)$ and $\mu_j^{(2)}(x,t,k)$ the first and second columns of $\mu_j(x,t;k)$, respectively.
\par
\begin{proposition}\label{pro1}
The matrices $\mu_1(x,t;k)$ and $\mu_2(x,t;k)$ have the folloeing properties:
\begin{enumerate}
\item $\det \mu_1(x,t,k)=\mu_2(x,t;k)=1$.
\item $\mu_1^{(1)}(x,t,k)$ is analytic in $\im k^2<0$ and
      \[
      \mu_1^{(1)}(x,t,k)=\left(\ba{c}1\\0\ea \right)+O(\frac{1}{k}),\mbox{as }k\rightarrow \infty,\quad \im k^2 \le 0.
      \]
\item $\mu_1^{(2)}(x,t,k)$ is analytic in $\im k^2>0$ and
      \[
      \mu_1^{(2)}(x,t,k)=\left(\ba{c}0\\1\ea \right)+O(\frac{1}{k}),\mbox{as }k\rightarrow \infty,\quad \im k^2 \geq 0.
      \]
\item $\mu_2^{(1)}(x,t,k)$ is analytic in $\im k^2>0$, and
      \[
      \mu_2^{(1)}(x,t,k)=\left(\ba{c}1\\0\ea \right)+O(\frac{1}{k}),\mbox{as }k\rightarrow \infty,\quad \im k^2 \geq 0.
      \]
\item $\mu_2^{(2)}(x,t,k)$ is analytic in $\im k^2<0$, and
      \[
      \mu_2^{(2)}(x,t,k)=\left(\ba{c}0\\1\ea \right)+O(\frac{1}{k}),\mbox{as }k\rightarrow \infty,\quad \im k^2 \le 0.
      \]
\item Moreover,
      \[
      \mu_j(x,t,k)=\id+\frac{\tilde \mu(x,t)}{ik}+o(\frac{1}{k})
      \]
      as $k\rightarrow \infty$ along curves transversal to the real and image axis, where
      \[
      [\sig_3,\tilde \mu(x,t)]=\left(\ba{cc}0&\tilde q(x,t)\\-\ol {\tilde q(x,t)}&0\ea \right)
      \]
      and
      \be \label{tldqdef}
      \tilde q(x,t)=q(x,t)e^{-2i\int_{+\infty}^{x}\Dta }
      \ee.
\end{enumerate}
\end{proposition}
Since the eigenfunctions $\mu_1(x,t,k)$ and $\mu_2(x,t,k)$ satisfy both equations of the Lax pair, we have
\be \label{psiphirel}
\mu_2(x,t,k)=\mu_1(x,t,k)S(k),\qquad k^2\in \R,
\ee
where $S(k)$ is independent of $(x,t)$ and is defined in (\ref{Skdefforab})

\begin{proposition}(Symmetries)
For $j=1,2$, the function $\mu(x,t,k)=\mu_j(x,t,k)$ satisfies the symmetry relations:
\be \label{sym}
\ba{c}
\mu_{11}(x,t,k)=\ol {\mu_{22}(x,t,\bar k)}, \\
\mu_{21}(x,t,k)=\ol {\mu_{12}(x,t,\bar k)},
\ea
\ee
as well as
\be \label{evenodd}
\ba{c}
\mu_{11}(x,t,-k)=\mu_{11}(x,t,k), \\
\mu_{12}(x,t,-k)=-\mu_{12}(x,t,k), \\
\mu_{21}(x,t,-k)=-\mu_{21}(x,t,k), \\
\mu_{22}(x,t,-k)=\mu_{22}(x,t,k).
\ea
\ee
\end{proposition}
Thus, we have
\be \label{Skdefforab}
S(k)=\left(\ba{cc}a(k)&-\ol{b(\bar k)}\\b(k)&\ol{a(\bar k)}\ea \right),
\ee

\section{The basic Riemann-Hilbert problem}

\subsection{The original Riemann-Hilbert problem}
The scattering relation (\ref{psiphirel}) involving the eigenfunctions $\Psi(x,t,k)=\mu_1(x,t,k)$ and $\Phi(x,t,k)=\mu_2(x,t,k)$ can be rewritten in the form of conjugation of boundary values of a piecewise analytic matrix-value function on a contour in the complex $k-$plane,namely:
\be \label{MMRHP}
M_+(x,t,k)=M_-(x,t,k)J(x,t,k),\qquad k^2\in \R,
\ee
where $M_\pm(x,t,k)$ denote the boundary vales of $M(x,t,k)$ according to a chosen orientation of $\Sig$, and $\Sig=\R \cup i\R$
\begin{figure}[th]
\centering
\includegraphics{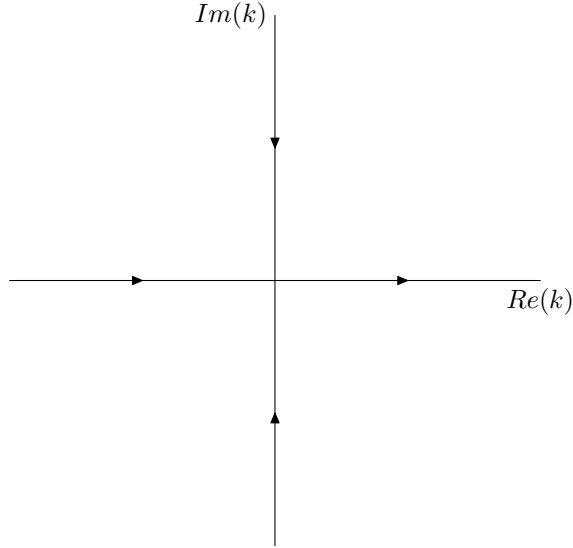}
\caption{The jump contour for $M$.}
\end{figure}
\par
Indeed,let us write (\ref{psiphirel}) in the vector form:
\be \label{rewritepsiphirel}
\begin{split}
&\frac{\Phi^{(1)}(x,t,k)}{a(k)}=\Psi^{(1)}(x,t,k)+r(k)\Psi^{(2)}(x,t,k),\\
&\frac{\Phi^{(2)}(x,t,k)}{\ol{a(\bar k)}}=-\ol{r(\bar k)}\Psi^{(1)}(x,t,k)+\Psi^{(2)}(x,t,k),
\end{split}
\ee
where
\be \label{rk}
r(k):=\frac{b(k)}{a(k)}
\ee
and define the matrix $M(x,t,k)$ as follows:
\be \label{Mdef}
M(x,t,k)=\left\{\ba{cc}(\ba{cc}\frac{\Phi^{(1)}(x,t,k)}{a(k)}e^{it\tha(k)}&\Psi^{(2)}(x,t,k)e^{-it\tha(k)}\ea),&k\in \{k\in \C|\im k^2>0\},\\(\ba{cc}\Psi^{(1)}(x,t,k)e^{it\tha(k)}&\frac{\Phi^{(2)}(x,t,k)}{\ol{a(\bar k)}}e^{-it\tha(k)}\ea),&k\in \{k\in \C|\im k^2<0\},\ea \right.
\ee
where
\be \label{thakdef}
\tha(k):=2k^4+\frac{x}{t}k^2,
\ee
Then the boundary values $M_+(x,t,k)$ and $M_-(x,t,k)$ relative to $\Sig$ are related by (\ref{MMRHP}),where
\be \label{jumpdef}
J(x,t,k)= \left(\ba{cc}1-r(k)\ol{r(\bar k)}&-\ol{r(\bar k)}e^{-2it\tha(k)}\\r(k)e^{2it\tha(k)}&1\ea \right),\qquad k^2\in \R,
\ee
\par
The jump relation (\ref{MMRHP}) considered together with the properties of the eigenfunctions listed in Proposition \ref{pro1} suggests a way of representing the solution to the Cauchy problem (\ref{DNLSandIv}) in terms of the solution of the Riemann-Hilbert problem, which is specified by the initial conditions (\ref{DNLSIv}) via the associated spectral function $r(k)$.
\par
The function $\tilde q(x,t)$ can be expressed in terms of the solution of the basic Riemann-Hilbert problem as follows:
\be \label{eqsol}
\tilde q(x,t)=2i\lim_{k \rightarrow \infty}(kM(x,t,k))_{12}.
\ee
where $M$ is the solution of the following Riemann-Hilbert problem:
\mbox{\bf The original Riemann-Hilbert problem:}
Given $r(k),k^2\in \R$, and $\Sig=\R \cup i\R$, find a $2\times 2$ matrix-value function $M(x,t,k)$ such that
\begin{enumerate}
\item $M(x,t,k)$ is analytic in $k\in \C\backslash \Sig$.
\item The boundary value $M_\pm(x,t,k)$ at $\Sig$ satisfy the jump condition
       \[
       M_+(x,t,k)=M_-(x,t,k)J(x,t,k),\quad k\in \Sig,
       \]
      where the jump matrix $J(x,t,k)$ is defined in terms of $r(k)$ by (\ref{jumpdef}).
\item Behavior at $\infty$
      \[
      M(x,t,k)=\id+O(\frac{1}{k}),\qquad \mbox{as }k\rightarrow \infty.
      \]
\end{enumerate}
And from the definition of the function $\tilde q(x,t)$ in (\ref{tldqdef}) we find
\be \label{qreltldq}
|q|=|\tilde q|
\ee
this means that the solution the the initial value problem (\ref{DNLSandIv}) can be expressed as follows:
\be\label{eqsolreal}
q(x,t)=\tilde q(x,t)e^{-2i\int_{+\infty}^{x}|\tilde q(y,t)|^2dy}.
\ee

\subsection{The new Riemann-Hilbert problem}\label{newRHPsec}

The jump condition (\ref{jumpdef}) is obtained in \cite{avkahv}. In that paper, the authors used this condition (\ref{jumpdef}) to analysis the long-time asymptotic behavior. But if we try to analysis the long-time asymptotic behavior of the DNLS equation (\ref{DNLS}) with step-like initial value problem, this type of Riemann-Hilbert problem has a contradiction in the plane wave region. So we try to derive a new Riemann-Hilbert problem, which is similar to the type of nonlinear Schr\"odinger equation, to overcome this contradiction. In this paper, we just analysis the long-time asymptotic behavior of the DNLS equation with decaying initial value problem. The step-like initial value problem will be obtained in another paper.
\par
We define
\be \label{Ndef}
\tilde N(x,t,k)=k^{-\frac{\hat \sig_3}{2}}M(x,t,k),
\ee
then the jump condition for $N$ is
\be \label{Njump}
\tilde N_+(x,t,k)=\tilde N_-(x,t,k)e^{-i(k^2x+2k^4t)\hat \sig_3}\left(\ba{cc}1-r(k)\ol{r(\bar k)}& -\frac{\ol{r(\bar k)}}{k}\\ k r(k)&1 \ea \right).
\ee
introducing $\lam=k^2$ and control the branch of $k$ as Sign $\im k=$Sign $\im \lam$, and define the modified scattering data $\rho(\lam)=\frac{r(k)}{k}$, \cite{kn}. And defining $N=B\tilde N$, where $B=\left(\ba{cc}1&0\\b&1\ea\right)$ and $b=\frac{\ol{\tilde q}}{2i}$,
\begin{figure}[th]
\centering
\includegraphics{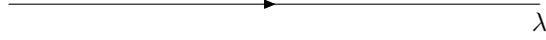}
\caption{The jump contour for $N$.}
\end{figure}
\par
Then the jump condition for $N$ is
\be \label{NJdef}
N_+(x,t,\lam)=
N_-(x,t,\lam)e^{-i(\lam x+2\lam^2t)\hat \sig_3}J_N(x,t,\lam).
\ee
where
\[
J_N(x,t,\lam)=\left(\ba{cc}1-\lam |\rho(\lam)|^2&- \bar \rho(\lam)\\ \lam \rho(\lam)&1\ea \right)
\]
the matrix $J_N$ admits the following triangular factorizations:
\be \label{JNfac}
\ba{c}
J_N=\left(\ba{cc}1&-\bar \rho(\lam)\\0&1\ea \right)\left(\ba{cc}1&0\\\lam \rho&1\ea \right)\\
=\left(\ba{cc}1&0\\\frac{\lam \rho}{1-\lam|\rho|^2}&1\ea \right)\left(\ba{cc}1-\lam|\rho|^2&0\\0&\frac{1}{1-\lam|\rho|^2}\ea\right)\left(\ba{cc}1&\frac{-\bar \rho}{1-\lam|\rho|^2}\\0&1\ea\right)
\ea
\ee
and the solution of the DNLS equation (\ref{DNLSandIv}) is
\begin{subequations}
\be\label{tldsol}
\tilde q(x,t)=2i\lim_{\lam \rightarrow \infty}(\lam N(x,t,\lam))_{12},
\ee
\be \label{solq}
q(x,t)=\tilde q(x,t)e^{-2i\int_{+\infty}^{x}|\tilde q(y,t)|^2dy}.
\ee
\end{subequations}

\section{Long-time analysis}

In order to analysis the long-time behavior of the solution of the DNLS equation, first we should split the jump matrix into an appropriate upper/lower triangular form, then this can help us localize the problem to the neighborhood of the stationary point. An appropriate rescaling then reduces the problem to a Riemann-Hilbert problem with constant coefficients, which can be solved explicitly in terms of classical functions.

\subsection{The first transformation} \label{first}

In this subsection we extend the Riemann-Hilbert problem (\ref{NJdef}) to an augmented contour of the type $\Sig_1$ given in Figure 4, which is constructed to reflect the signature of $\re it\tha$. For technical reasons we will assume that the contour $\Sig_1$ is composed of straight lines, as shown in Figure 4 below with angle $\frac{\pi}{4}$, although any contour of the same general shape as $\Sig_1$ would do.
\begin{figure}[th]
\centering
\includegraphics{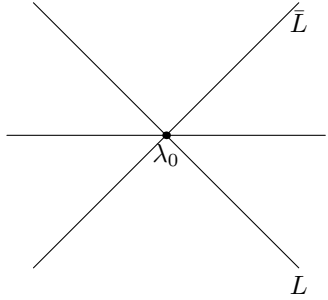}
\caption{The contour of $\Sig_1$.}
\end{figure}
\par
As in the \cite{dz}, we first consider the stationary point of the function $\tha(\lam)=2\lam^2+\frac{x}{t}\lam$, that is, letting
\[
\frac{d(\tha(\lam))}{d\lam}=0
\]
we get the stationary point $\lam_0=\frac{-x}{4t}$.
\par
And we also get the signature table of $\re \tha(\lam)$ that is as follows in Fig.5.
\begin{figure}[th]
\centering
\includegraphics{IST2.1}
\caption{The signature table of $\re (i\tha(\lam))$.}
\end{figure}
\par
Then we introduce a scalar function. Let $\dta(\lam)$ be the solution of the scalar factorization problem
\be \label{dtaRHP}
\left\{
\ba{ll}
\dta_+(\lam)=\dta_-(\lam)(1-\lam |\rho(\lam)|^2),&\lam<\lam_0,\\
\dta_+(\lam)=\dta_-(\lam),&\lam>\lam_0,\\
\dta(\lam)\rightarrow 1,&\lam \rightarrow \infty.
\ea
\right.
\ee
Direct calculation shows that (\ref{dtaRHP}) is solved by the formula
\be \label{dtasol}
\dta(\lam)=\exp{\frac{1}{2\pi i}}\int_{-\infty}^{\lam_0}\frac{\log{(1-\lam' |\rho(\lam')|^2)}}{\lam'-\lam}d\lam',
\ee
\par
And we can find that $\dta(\lam)$ and $\dta(\lam)^{-1}$ are uniformly bounded in $\lam$ and for $|\lam_0|\le M$.
\par
We conjugate the Riemann-Hilbert problem (\ref{NJdef}) by
\be \label{condta}
\dta^{\sig_3}(\lam)=\left(\ba{cc}\dta(\lam)&0\\0&\dta^{-1}(\lam)\ea \right),
\ee
leads to the factorization problem $N^{(1)}(x,t,\lam)=N(x,t,\lam)\dta^{-\sig_3}(\lam)$,
\be \label{N1rhp}
\left\{
\ba{cl}
N^{(1)}_{+}(x,t,\lam)=N^{(1)}_{-}(x,t,\lam)J_{N^{(1)}}(x,t,\lam),&\lam \in \R,\\
N^{(1)}(x,t,\lam)\rightarrow \id,&\lam \rightarrow \infty.
\ea
\right.
\ee
where
\be \label{N1jump}
J_{N^{(1)}}(x,t,\lam)=\dta_-^{\sig_3}J_{N^{(1)}}(x,t,\lam)\dta_+^{-\sig_3}(\lam)=
\left\{\ba{ll}e^{-it\dta \hat \sig_3}\left(\ba{cc}1&0\\\frac{\lam \rho\dta_-^{-2}}{1-\lam|\rho|^2}&1\ea \right)\left(\ba{cc}1&\frac{-\bar \rho\dta_+^{2}}{1-\lam|\rho|^2}\\0&1\ea\right),&\lam<\lam_0\\
e^{-it\dta \hat \sig_3}\left(\ba{cc}1&-\bar \rho\dta^2\\0&1\ea \right)\left(\ba{cc}1&0\\\lam \rho\dta^{-2}&1\ea \right),&\lam>\lam_0
\ea
\right.
\ee
\par
Having made the above definitions, we now describe the strategy. Suppose that the coefficients
\be \label{JN1coff}
\frac{\lam \rho}{1-\lam|\rho|^2},\quad \frac{\bar \rho}{1-\lam|\rho|^2},\quad \bar \rho,\quad \lam \rho
\ee
can be replaced by some rational functions
\be \label{JN1coffrat}
[\frac{\lam \rho}{1-\lam|\rho|^2}],\quad [\frac{\bar \rho}{1-\lam|\rho|^2}],\quad [\bar \rho],\quad [\lam \rho]
\ee
respectively. Then if the poles of these functions are appropriately placed, the Riemann-Hilbert problem on $\R$
can be deformed to the contour $\Sig_1$.
\begin{remark}
In this paper, we assume that the function $\rho(\lam)$ has no zero.
\end{remark}
\par
To verify that the coefficients (\ref{JN1coff}) can be replaced by the rational functions (\ref{JN1coffrat}) with well-controlled errors, we proceed as follows.
\par
\begin{description} \label{case}
\item[1] For $\lam<\lam_0$, $\frac{-\bar \rho}{1-\lam|\rho|^2}$ \label{case1}
\par
Set
\be \label{f1def}
f_1(\lam)=\frac{-\bar \rho}{1-\lam|\rho|^2}
\ee
By Taylor's formula, we have
\be \label{f1taylor}
(\lam-i)^{m+5}f(\lam)=\mu_0+\mu_1(\lam-\lam_0)+\cdots +\mu_m(\lam-\lam_0)^m+\frac{1}{m!}\int_{\lam_0}^{\lam}((\cdot -i)^{m+5}f(\cdot))^{(m+1)}(\gam)(\lam-\gam)^md\gam,
\ee
and define
\be \label{f1R}
R(\lam)=\frac{\sum_{i=0}^{m}\mu_i(\lam-\lam_0)^i}{(\lam-i)^{m+5}},
\ee
\be \label{f1h}
h(\lam)=f(\lam)-R(\lam),
\ee
As before, the proof of the following result is straightforward:
\begin{lemma}
\be \label{dfreldR}
\frac{d^jf(\lam)}{d\lam^j}|_{\lam_0}=\frac{d^jR(\lam)}{d\lam^j}|_{\lam_0},\quad 0\le j\le m.
\ee
Also, $\mu_i=\mu_i(\lam_0)$ decays rapidly as $\lam_0\rightarrow \infty$.
\end{lemma}
\begin{proof}
Formula (\ref{dfreldR}) is immediate. The decay as $\lam_0\rightarrow \infty$ follows from the formulae
\[
\mu_i=\frac{1}{i!}\frac{d^if(u)}{du^i}|_{\lam_0}.
\]
\end{proof}
\par
In what follows we fix $m\in \Z_+$ and, for convenience, we assume that $m$ is of the form
\be \label{kform}
k=4q+1,\quad q\in \Z_+,
\ee
Write
\be \label{f1hr}
f(\lam)=h(\lam)+R(\lam),\quad \lam<\lam_0,
\ee
Then by the above lemma,
\be \label{dh}
\frac{d^jh(\lam)}{d\lam^j}|_{\lam_0}=0,\quad 0\le j\le m,
\ee
We now use this property to split $h$ further as
\be \label{hsplit}
h(\lam)=h_1(\lam)+h_2(\lam),
\ee
where $h_1(\lam)$ is small and $h_2(\lam)$ has an analytic continuation to $\lam+i0$. Thus
\be
f=h_1+(h_2+R)
\ee
is the desired splitting of $f$.
\par
Set
\be \label{f1bta}
\beta(\lam)=\frac{(\lam-\lam_0)^{q}}{(\lam-i)^{q+2}}.
\ee
Consider the Fourier transform with respect to $\tha$. As $\lam \rightarrow \tha(\lam)$ is one-to-one in $\lam<\lam_0$,we define
\be
\left\{
\ba{lll}
\frac{h}{\beta}(\tha)&=\frac{h(\lam(\tha))}{\beta(\lam(\tha))},&\tha>\tha(\lam_0)=-2\lam_0^2,\\
&=0,&\tha\le-2\lam_0^2
\ea
\right.
\ee
Thus, as $\tha>-2\lam_0^2$, from formulae (\ref{f1taylor}), (\ref{f1h}) and (\ref{f1bta}) it follows that
\be \label{f1hdef}
\frac{h}{\beta}(\lam)=\frac{(\lam-\lam_0)^{m+1-q}}{(\lam-i)^{m+3-q}}g(\lam,\lam_0),
\ee
where
\be
g(\lam,\lam_0)=\frac{1}{m!}\int_0^1((\cdot -i)^{m+5}f(\cdot))^{(m+1)}(\lam_0+u(\lam-\lam_0))(1-u)^mdu.
\ee
from which we see that
\be
\left |\frac{d^jg(\lam,\lam_0)}{d\lam^j}\right |\le C,\quad \lam\le \lam_0.
\ee
Then, we obtain
\be \label{halup}
\ba{rl}
\int_{-\infty}^{\infty}\left |(\frac{d}{d\tha})^j(\frac{h}{\beta}(\lam(\tha)))\right |^2\bar d\tha=&\int_{-\infty}^{\lam_0}\left |(\frac{1}{4(\lam-\lam_0)}\frac{d}{d\lam})^j(\frac{h}{\beta}(\lam))\right |^24(\lam-\lam_0)\bar d\lam\\
\le&C\int_{-\infty}^{\lam_0}\left | \frac{(\lam-\lam_0)^{m+1-q-2j}}{(\lam-i)^{m+3-q}}\right |^2(\lam-\lam_0)\bar d\lam\le C\\
&\mbox{for }0\le j\le \frac{m+1-q}{2}=\frac{3q+2}{2}.
\ea
\ee
By Plancherel,
\be \label{hathalup}
\int_{-\infty}^{\infty}(1+s^2)^j|\widehat{(h/\beta)}(s)|^2ds\le C<\infty,\quad 0\le j\le \frac{3q+2}{2}.
\ee
where
\be
\widehat{(h/\beta)}(s)=\int_{-\infty}^{\lam_0}e^{-is\tha(\lam)}(h/\beta)(\lam)\bar d\tha(\lam),
\ee
And by Fourier,
\be
(h/\beta)(\lam)=\int_{-\infty}^{\infty}e^{is\tha(\lam)}\widehat{(h/\beta)}(s)\bar ds.
\ee
In the above formulae we use the convenient notation $\bar ds=\frac{ds}{\sqrt{2\pi}}$ and $\bar d(\tha(\lam))=\frac{d\tha(\lam)}{\sqrt{2\pi}}$.
\par
\begin{remark}
The constants in (\ref{halup}) and (\ref{hathalup}) should properly be denoted by $c_1$ and $c_2$. Here, and in what follows, we use $c$ and sometimes $C$ to denote a generic constant. This abuse of notation should not give rise to any confusion.
\end{remark}
We split
\be \label{h11split}
\ba{rl}
h(\lam)&=\beta(\lam)\int_t^{\infty}e^{is\tha(\lam)}\widehat{(h/\beta)}(s)\bar ds+\beta(\lam)\int_{-\infty}^te^{is\tha(\lam)}\widehat{(h/\beta)}(s)\bar ds\\
&=h_1(\lam)+h_2(\lam).
\ea
\ee
For $\lam\le\lam_0$ we find that
\be
|e^{-2it\tha(\lam)}h_1(\lam)|\le\frac{C}{|\lam-i|^2t^{p-\frac{1}{2}}},\quad \mbox{for any }p\le\frac{3q+2}{2}.
\ee
\par
On the other hand, $h_2(\lam)$ has an analytic continuation to the upper half-plane, where $\re{i\tha(\lam)}>0$, and for $\lam$ on the line $\lam_0+\lam_0ue^{i\frac{3\pi}{4}},u\ge 0$,
\be
|e^{-2it\tha(\lam)}h_2(\lam)|\le\frac{c\lam_0^qu^qe^{-t\re{i\tha(\lam)}}}{|\lam-i|^{q+2}}.
\ee
However, from expression of $\tha(\lam)$, that is
\be
\tha(\lam)=2(\lam-\lam_0)^2-2\lam_0^2,
\ee
we have
\be
\re{i\tha(\lam)}=2\lam_0^2u^2,
\ee
and hence
\be
\ba{rl}
|e^{-2it\tha(\lam)}h_2(\lam)|&\le\frac{c\lam_0^q[((t\lam_0^2)^{\frac{1}{2}}u)^qe^{-2t\lam_0^2u^2}]}{|\lam-i|^{q+2}(t\lam_0^2)^{\frac{q}{2}}}\\
&\le\frac{c}{|\lam-i|^{q+2}t^{\frac{q}{2}}}\le\frac{c}{|\lam-i|^2t^{\frac{q}{2}}}.
\ea
\ee
On the line $\lam_0+\lam_0ue^{i\frac{3\pi}{4}},u\ge \eps,\eps>0$ we have
\be
|e^{-2it\tha(\lam)}R(\lam)|\le Ce^{-4t\lam_0^2u^2}\le Ce^{-4t\lam_0^2\eps^2}
\ee

\par

\item[2] For case: $\lam<\lam_0$, $\frac{\lam \rho}{1-\lam|\rho|^2}$ \label{case2}
\par
Set
\be
f_2(\lam)=\frac{\lam \rho}{1-\lam|\rho|^2}.
\ee
Again, by Taylor's formula, we have
\be \label{f2taylor}
(\lam+i)^{m+5}f(\lam)=\mu_0+\mu_1(\lam-\lam_0)+\cdots +\mu_m(\lam-\lam_0)^m+\frac{1}{m!}\int_{\lam_0}^{\lam}((\cdot +i)^{m+5}f(\cdot))^{(m+1)}(\gam)(\lam-\gam)^md\gam,
\ee
and define
\be \label{f2R}
R(\lam)=\frac{\sum_{i=0}^{m}\mu_i(\lam-\lam_0)^i}{(\lam+i)^{m+5}},
\ee
\be \label{f2h}
h(\lam)=f(\lam)-R(\lam),
\ee
As before, the proof of the following result is straightforward:
\begin{lemma}
\be
\frac{d^jf(\lam)}{d\lam^j}|_{\lam_0}=\frac{d^jR(\lam)}{d\lam^j}|_{\lam_0},\quad 0\le j\le m.
\ee
Also, $\mu_i=\mu_i(\lam_0)$ decays rapidly as $\lam_0\rightarrow \infty$.
\end{lemma}
Clearly $\frac{d^jh(\lam)}{d\lam^j}|_{\lam_0}=0,0\le j\le m$.
\par
Set
\be \label{f2bta}
\beta(\lam)=\frac{(\lam-\lam_0)^{q}}{(\lam+i)^{q+2}}.
\ee
From formulae (\ref{f2taylor}), (\ref{f2h}) and (\ref{f2bta}) it follows that
\be \label{f2hdef}
\frac{h}{\beta}(\lam)=\frac{(\lam-\lam_0)^{m+1-q}}{(\lam+i)^{m+3-q}}g(\lam,\lam_0),
\ee
where
\be
g(\lam,\lam_0)=\frac{1}{m!}\int_0^1((\cdot +i)^{m+5}f(\cdot))^{(m+1)}(\lam_0+u(\lam-\lam_0))(1-u)^mdu.
\ee
from which we see that
\be
\left |\frac{d^jg(\lam,\lam_0)}{d\lam^j}\right |\le C,\quad \lam\le \lam_0.
\ee
Then, we obtain
\be
\ba{rl}
\int_{-\infty}^{\infty}\left |(\frac{d}{d\tha})^j(\frac{h}{\beta}(\lam(\tha)))\right |^2\bar d\tha=&\int_{-\infty}^{\lam_0}\left |(\frac{1}{4(\lam-\lam_0)}\frac{d}{d\lam})^j(\frac{h}{\beta}(\lam))\right |^24(\lam-\lam_0)\bar d\lam\\
\le&C\int_{-\infty}^{\lam_0}\left | \frac{(\lam-\lam_0)^{m+1-q-2j}}{(\lam+i)^{m+3-q}}\right |^2(\lam-\lam_0)\bar d\lam\le C\\
&\mbox{for }0\le j\le \frac{m+1-q}{2}=\frac{3q+2}{2}.
\ea
\ee
By Plancherel,
\be
\int_{-\infty}^{\infty}(1+s^2)^j|\widehat{(h/\beta)}(s)|^2ds\le C<\infty,\quad 0\le j\le \frac{3q+2}{2}.
\ee
where
\be
\widehat{(h/\beta)}(s)=\int_{-\infty}^{\lam_0}e^{is\tha(\lam)}(h/\beta)(\lam)\bar d\tha(\lam),
\ee
And by Fourier,
\be
(h/\beta)(\lam)=\int_{-\infty}^{\infty}e^{-is\tha(\lam)}\widehat{(h/\beta)}(s)\bar ds.
\ee
Again we split
\be
\ba{rl}
h(\lam)&=\beta(\lam)\int_t^{\infty}e^{-is\tha(\lam)}\widehat{(h/\beta)}(s)\bar ds+\beta(\lam)\int_{-\infty}^te^{-is\tha(\lam)}\widehat{(h/\beta)}(s)\bar ds\\
&=h_1(\lam)+h_2(\lam).
\ea
\ee
For $\lam\le\lam_0$ we find that
\be
|e^{2it\tha(\lam)}h_1(\lam)|\le\frac{C}{|\lam+i|^2t^{p-\frac{1}{2}}},\quad \mbox{for any }p\le\frac{3q+2}{2}.
\ee
\par
On the other hand, $h_2(\lam)$ has an analytic continuation to the lower half-plane, where $\re{i\tha(\lam)}<0$, and for $\lam$ on the line $\lam_0+\lam_0ue^{-i\frac{3\pi}{4}},u\ge 0$,
\be
|e^{2it\tha(\lam)}h_2(\lam)|\le\frac{c\lam_0^qu^qe^{t\re{i\tha(\lam)}}}{|\lam+i|^{q+2}}.
\ee
However, from expression of $\tha(\lam)$, that is
\be
\tha(\lam)=2(\lam-\lam_0)^2-2\lam_0^2,
\ee
we have
\be
\re{i\tha(\lam)}=-2\lam_0^2u^2,
\ee
and hence
\be
\ba{rl}
|e^{2it\tha(\lam)}h_2(\lam)|&\le\frac{c\lam_0^q[((t\lam_0^2)^{\frac{1}{2}}u)^qe^{-2t\lam_0^2u^2}]}{|\lam+i|^{q+2}(t\lam_0^2)^{\frac{q}{2}}}\\
&\le\frac{c}{|\lam+i|^{q+2}t^{\frac{q}{2}}}\le\frac{c}{|\lam+i|^2t^{\frac{q}{2}}}.
\ea
\ee
On the line $\lam_0+\lam_0ue^{-i\frac{3\pi}{4}},u\ge \eps,\eps>0$ we have
\be
|e^{-2it\tha(\lam)}R(\lam)|\le Ce^{-4t\lam_0^2u^2}\le Ce^{-4t\lam_0^2\eps^2}
\ee
\par
In fact this case is just the conjugate of the above case. And the two cases in the following is fimilar with these two cases, but we write them down here for the reader's convenience.

\par
\item[3] For case: $\lam>\lam_0$, $\ol{\rho(\lam)}$ \label{case3}
\par
Set
\be \label{f3}
f_3(\lam)=\ol{\rho(\lam)},\quad \lam\ge \lam_0,
\ee
Again, by Taylor's formula, we have
\be \label{f3taylor}
(\lam-i)^{m+5}f(\lam)=\mu_0+\mu_1(\lam-\lam_0)+\cdots +\mu_m(\lam-\lam_0)^m+\frac{1}{m!}\int_{\lam_0}^{\lam}((\cdot -i)^{m+5}f(\cdot))^{(m+1)}(\gam)(\lam-\gam)^md\gam,
\ee
and define
\be \label{f3R}
R(\lam)=\frac{\sum_{i=0}^{m}\mu_i(\lam-\lam_0)^i}{(\lam-i)^{m+5}},
\ee
\be \label{f3h}
h(\lam)=f(\lam)-R(\lam),
\ee
As before, the proof of the following result is straightforward:
\begin{lemma}
\be
\frac{d^jf(\lam)}{d\lam^j}|_{\lam_0}=\frac{d^jR(\lam)}{d\lam^j}|_{\lam_0},\quad 0\le j\le m.
\ee
Also, $\mu_i=\mu_i(\lam_0)$ decays rapidly as $\lam_0\rightarrow \infty$.
\end{lemma}
Clearly $\frac{d^jh(\lam)}{d\lam^j}|_{\lam_0}=0,0\le j\le m$.
\par
Set
\be \label{f3bta}
\beta(\lam)=\frac{(\lam-\lam_0)^{q}}{(\lam-i)^{q+2}}.
\ee
From formulae (\ref{f3taylor}), (\ref{f3h}) and (\ref{f3bta}) it follows that
\be \label{f3hdef}
\frac{h}{\beta}(\lam)=\frac{(\lam-\lam_0)^{m+1-q}}{(\lam-i)^{m+3-q}}g(\lam,\lam_0),
\ee
where
\be
g(\lam,\lam_0)=\frac{1}{m!}\int_0^1((\cdot -i)^{m+5}f(\cdot))^{(m+1)}(\lam_0+u(\lam-\lam_0))(1-u)^mdu.
\ee
from which we see that
\be
\left |\frac{d^jg(\lam,\lam_0)}{d\lam^j}\right |\le C,\quad \lam\ge \lam_0.
\ee
Then, we obtain
\be
\ba{rl}
\int_{-\infty}^{\infty}\left |(\frac{d}{d\tha})^j(\frac{h}{\beta}(\lam(\tha)))\right |^2\bar d\tha=&\int_{\lam_0}^{\infty}\left |(\frac{1}{4(\lam-\lam_0)}\frac{d}{d\lam})^j(\frac{h}{\beta}(\lam))\right |^24(\lam-\lam_0)\bar d\lam\\
\le&C\int_{\lam_0}^{\infty}\left | \frac{(\lam-\lam_0)^{m+1-q-2j}}{(\lam-i)^{m+3-q}}\right |^2(\lam-\lam_0)\bar d\lam\le C\\
&\mbox{for }0\le j\le \frac{m+1-q}{2}=\frac{3q+2}{2}.
\ea
\ee
By Plancherel,
\be
\int_{-\infty}^{\infty}(1+s^2)^j|\widehat{(h/\beta)}(s)|^2ds\le C<\infty,\quad 0\le j\le \frac{3q+2}{2}.
\ee
where
\be
\widehat{(h/\beta)}(s)=\int_{\lam_0}^{\infty}e^{-is\tha(\lam)}(h/\beta)(\lam)\bar d\tha(\lam),
\ee
And by Fourier,
\be
(h/\beta)(\lam)=\int_{-\infty}^{\infty}e^{is\tha(\lam)}\widehat{(h/\beta)}(s)\bar ds.
\ee
Again we split
\be
\ba{rl}
h(\lam)&=\beta(\lam)\int_t^{\infty}e^{is\tha(\lam)}\widehat{(h/\beta)}(s)\bar ds+\beta(\lam)\int_{-\infty}^te^{is\tha(\lam)}\widehat{(h/\beta)}(s)\bar ds\\
&=h_1(\lam)+h_2(\lam).
\ea
\ee
For $\lam\ge\lam_0$ we find that
\be
|e^{-2it\tha(\lam)}h_1(\lam)|\le\frac{C}{|\lam-i|^2t^{p-\frac{1}{2}}},\quad \mbox{for any }p\le\frac{3q+2}{2}.
\ee
\par
On the other hand, $h_2(\lam)$ has an analytic continuation to the lower half-plane, where $\re{i\tha(\lam)}>0$, and for $\lam$ on the line $\lam_0+\lam_0ue^{-i\frac{\pi}{4}},u\ge 0$,
\be
|e^{-2it\tha(\lam)}h_2(\lam)|\le\frac{c\lam_0^qu^qe^{-t\re{i\tha(\lam)}}}{|\lam-i|^{q+2}}.
\ee
However, from expression of $\tha(\lam)$, that is
\be
\tha(\lam)=2(\lam-\lam_0)^2-2\lam_0^2,
\ee
we have
\be
\re{i\tha(\lam)}=2\lam_0^2u^2,
\ee
and hence
\be
\ba{rl}
|e^{-2it\tha(\lam)}h_2(\lam)|&\le\frac{c\lam_0^q[((t\lam_0^2)^{\frac{1}{2}}u)^qe^{-2t\lam_0^2u^2}]}{|\lam-i|^{q+2}(t\lam_0^2)^{\frac{q}{2}}}\\
&\le\frac{c}{|\lam-i|^{q+2}t^{\frac{q}{2}}}\le\frac{c}{|\lam-i|^2t^{\frac{q}{2}}}.
\ea
\ee
On the line $\lam_0+\lam_0ue^{-i\frac{\pi}{4}},u\ge \eps,\eps>0$ we have
\be
|e^{-2it\tha(\lam)}R(\lam)|\le Ce^{-4t\lam_0^2u^2}\le Ce^{-4t\lam_0^2\eps^2}
\ee

\par
\item[4] For case: $\lam\ge\lam_0:$ $\lam\rho(\lam)$   \label{case4}
\par
Set
\be
f(\lam)=\lam\rho(\lam).
\ee
Again, by Taylor's formula, we have
\be \label{f4taylor}
(\lam+i)^{k+5}f(\lam)=\mu_0+\mu_1(\lam-\lam_0)+\cdots +\mu_m(\lam-\lam_0)^k+\frac{1}{m!}\int_{\lam_0}^{\lam}((\cdot +i)^{k+5}f(\cdot))^{(m+1)}(\gam)(\lam-\gam)^md\gam,
\ee
and define
\be \label{f4R}
R(\lam)=\frac{\sum_{i=0}^{m}\mu_i(\lam-\lam_0)^i}{(\lam+i)^{m+5}},
\ee
\be \label{f4h}
h(\lam)=f(\lam)-R(\lam),
\ee
As before, the proof of the following result is straightforward:
\begin{lemma}
\be
\frac{d^jf(\lam)}{d\lam^j}|_{\lam_0}=\frac{d^jR(\lam)}{d\lam^j}|_{\lam_0},\quad 0\le j\le m.
\ee
Also, $\mu_i=\mu_i(\lam_0)$ decays rapidly as $\lam_0\rightarrow \infty$.
\end{lemma}
Clearly $\frac{d^jh(\lam)}{d\lam^j}|_{\lam_0}=0,0\le j\le m$.
\par
Set
\be \label{f4bta}
\beta(\lam)=\frac{(\lam-\lam_0)^{q}}{(\lam+i)^{q+2}}.
\ee
From formulae (\ref{f4taylor}), (\ref{f4h}) and (\ref{f4bta}) it follows that
\be \label{f4hdef}
\frac{h}{\beta}(\lam)=\frac{(\lam-\lam_0)^{m+1-q}}{(\lam+i)^{m+3-q}}g(\lam,\lam_0),
\ee
where
\be
g(\lam,\lam_0)=\frac{1}{m!}\int_0^1((\cdot +i)^{m+5}f(\cdot))^{(m+1)}(\lam_0+u(\lam-\lam_0))(1-u)^mdu.
\ee
from which we see that
\be
\left |\frac{d^jg(\lam,\lam_0)}{d\lam^j}\right |\le C,\quad \lam\ge \lam_0.
\ee
Then, we obtain
\be
\ba{rl}
\int_{-\infty}^{\infty}\left |(\frac{d}{d\tha})^j(\frac{h}{\beta}(\lam(\tha)))\right |^2\bar d\tha=&\int_{\lam_0}^{\infty}\left |(\frac{1}{4(\lam-\lam_0)}\frac{d}{d\lam})^j(\frac{h}{\beta}(\lam))\right |^24(\lam-\lam_0)\bar d\lam\\
\le&C\int_{\lam_0}^{\infty}\left | \frac{(\lam-\lam_0)^{m+1-q-2j}}{(\lam+i)^{m+3-q}}\right |^2(\lam-\lam_0)\bar d\lam\le C\\
&\mbox{for }0\le j\le \frac{k+1-q}{2}=\frac{3q+2}{2}.
\ea
\ee
By Plancherel,
\be
\int_{-\infty}^{\infty}(1+s^2)^j|\widehat{(h/\beta)}(s)|^2ds\le C<\infty,\quad 0\le j\le \frac{3q+2}{2}.
\ee
where
\be
\widehat{(h/\beta)}(s)=\int_{\lam_0}^{\infty}e^{is\tha(\lam)}(h/\beta)(\lam)\bar d\tha(\lam),
\ee
And by Fourier,
\be
(h/\beta)(\lam)=\int_{-\infty}^{\infty}e^{-is\tha(\lam)}\widehat{(h/\beta)}(s)\bar ds.
\ee
Again we split
\be
\ba{rl}
h(\lam)&=\beta(\lam)\int_t^{\infty}e^{-is\tha(\lam)}\widehat{(h/\beta)}(s)\bar ds+\beta(\lam)\int_{-\infty}^te^{-is\tha(\lam)}\widehat{(h/\beta)}(s)\bar ds\\
&=h_1(\lam)+h_2(\lam).
\ea
\ee
For $\lam\ge\lam_0$ we find that
\be
|e^{2it\tha(\lam)}h_1(\lam)|\le\frac{C}{|\lam+i|^2t^{p-\frac{1}{2}}},\quad \mbox{for any }p\le\frac{k+1-q}{2}.
\ee
\par
On the other hand, $h_2(\lam)$ has an analytic continuation to the upper half-plane, where $\re{i\tha(\lam)}<0$, and for $\lam$ on the line $\lam_0+\lam_0ue^{i\frac{\pi}{4}},u\ge 0$,
\be
|e^{2it\tha(\lam)}h_2(\lam)|\le\frac{c\lam_0^qu^qe^{t\re{i\tha(\lam)}}}{|\lam+i|^{q+2}}.
\ee
However, from expression of $\tha(\lam)$, that is
\be
\tha(\lam)=2(\lam-\lam_0)^2-2\lam_0^2,
\ee
we have
\be
\re{i\tha(\lam)}=-2\lam_0^2u^2,
\ee
and hence
\be
\ba{rl}
|e^{2it\tha(\lam)}h_2(\lam)|&\le\frac{c\lam_0^q[((t\lam_0^2)^{\frac{1}{2}}u)^qe^{-2t\lam_0^2u^2}]}{|\lam+i|^{q+2}(t\lam_0^2)^{\frac{q}{2}}}\\
&\le\frac{c}{|\lam+i|^{q+2}t^{\frac{q}{2}}}\le\frac{c}{|\lam+i|^2t^{\frac{q}{2}}}.
\ea
\ee
On the line $\lam_0+\lam_0ue^{i\frac{\pi}{4}},u\ge \eps,\eps>0$ we have
\be
|e^{-2it\tha(\lam)}R(\lam)|\le Ce^{-4t\lam_0^2u^2}\le Ce^{-4t\lam_0^2\eps^2}
\ee
\end{description}

\par
We can summarize our results as follows: let $l$ be an arbitrary positive integer and let $k=4q+1$ be sufficiently large that the integers that are the last formula of the above formulas about $h_2$ are all greater than $l$. Let $L$ denote the contour
\be \label{Lcon}
L:\{\lam=\lam_0+\lam_0ue^{-i\frac{\pi}{4}};u\ge 0\}\cup \{\lam=\lam_0+\lam_0ue^{i\frac{3\pi}{4}};u\ge 0\}
\ee
so that the contour $\Sig_1$ in Figure 4 is given by
\be
\Sig_1=L\cup \bar L\cup \R.
\ee
Also set
\be
L_{\eps}=\{\lam=\lam_0+\lam_0ue^{i\frac{3\pi}{4}},u\ge \eps \}\cup\{\lam=\lam_0+\lam_0ue^{-i\frac{\pi}{4}},u\ge \eps \}
\ee
\begin{proposition} \label{summtec}
Let
\be
f(\lam)=
\left\{
\ba{lll}
f_1(\lam)=&\frac{-\bar \rho}{1-\lam|\rho|^2},& \lam<\lam_0\\
f_2(\lam)=&\frac{\lam \rho}{1-\lam|\rho|^2},& \lam<\lam_0\\
f_3(\lam)=&\ol{\rho(\lam)},&\lam>\lam_0\\
f_4(\lam)=&\lam\rho(\lam),& \lam>\lam_0
\ea
\right.
\ee
Then $f$ has a decomposition
\be
f(\lam)=h_1(\lam)+h_2(\lam)+R(\lam),\quad \lam \in \R,
\ee
where $R(\lam)$ is piecewise rational and $h_2(\lam)$ has an analytic continuation to $L$ or $\bar L$ satisfying
for case (\ref{case1}) and (\ref{case3})
\be
|e^{-2it\tha(\lam)}h_1(\lam)|\le \frac{c}{(1+|\lam|^2)t^l},\quad \lam \in \R,
\ee
\be
|e^{-2it\tha(\lam)}h_2(\lam)|\le \frac{c}{(1+|\lam|^2)t^l},\quad \lam \in L,
\ee
and
\be
|e^{-2it\tha(\lam)}R(\lam)|\le Ce^{-4t\lam_0^2u^2}\le Ce^{-4\eps^2\tau},\quad \lam \in L_{\eps},
\ee
or
for case (\ref{case2}) and (\ref{case4})
\be
|e^{2it\tha(\lam)}h_1(\lam)|\le \frac{c}{(1+|\lam|^2)t^l},\quad \lam \in \R,
\ee
\be
|e^{2it\tha(\lam)}h_2(\lam)|\le \frac{c}{(1+|\lam|^2)t^l},\quad \lam \in \bar L,
\ee
and
\be
|e^{2it\tha(\lam)}R(\lam)|\le Ce^{-4t\lam_0^2u^2}\le Ce^{-4\eps^2\tau},\quad \lam \in \bar L_{\eps},
\ee
where $l$ is an arbitary positive integer and $\tau=t\lam_0^2$.
\end{proposition}
\par
Finally we extend the Riemann-Hilbert problem (\ref{Njump}) to the augmente contour $\Sig_1$ of Figure 4. From problem (\ref{N1rhp}) and formulae(\ref{N1jump}), the Riemann-Hilbert problem across $\R$ oriented as Figure 3 is given by
\be \label{N1RHP}
\begin{split}
N^{(1)}_+=N^{(1)}_-(b_-)_{x,t,\dta}^{-1}(b_+)_{x,t,\dta}&\\
N^{(1)}\rightarrow \id,\qquad \lam \rightarrow \infty.&
\end{split}
\ee
where
\be
(b_{\pm})_{x,t,\dta}=\dta_{\pm}^{\hat \sig_3}e^{-it\tha(\lam)\hat \sig_3}b_{\pm},
\ee
\be
b_+=\id+w_+=\left\{
\ba{ll}
\left(\ba{cc}1&f_{1}\\0&1\ea \right),&\lam<\lam_0\\
\left(\ba{cc}1&f_{3}\\0&1\ea \right),&\lam<\lam_0
\ea
\right.
\ee
\be
b_-=\id+w_-=\left\{
\ba{ll}
\left(\ba{cc}1&f_{2}\\0&1\ea \right),&\lam<\lam_0\\
\left(\ba{cc}1&f_{4}\\0&1\ea \right),&\lam<\lam_0
\ea
\right.
\ee
\par
Orient $\Sig_1$ in Figure 4 as in Figure 6 and write
\begin{subequations}\label{bpmdef}
\be \label{bpdef}
b_+=b_+^ob_+^a=(\id+w_+^o)(\id+w_+^a)=\left\{
\ba{ll}
\left(\ba{cc}1&(h_1)^{1}\\0&1\ea\right)\left(\ba{cc}1&(h_2+R)^{1}\\0&1\ea\right),&\lam<\lam_0\\
\left(\ba{cc}1&(h_1)^{3}\\0&1\ea\right)\left(\ba{cc}1&(h_2+R)^{3}\\0&1\ea\right),&\lam>\lam_0
\ea
\right.
\ee
\be\label{bmdef}
b_-=b_-^ob_-^a=(\id+w_-^o)(\id+w_-^a)=\left\{
\ba{ll}
\left(\ba{cc}1&(h_1)^{2}\\0&1\ea\right)\left(\ba{cc}1&(h_2+R)^{2}\\0&1\ea\right),&\lam<\lam_0\\
\left(\ba{cc}1&(h_1)^{4}\\0&1\ea\right)\left(\ba{cc}1&(h_2+R)^{4}\\0&1\ea\right),&\lam>\lam_0
\ea
\right.
\ee
\end{subequations}
where the function $(h)^{j}$ denotes the function is defined in the case $j$, $j=1,2,3,4$.
\begin{figure}[th]
\centering
\includegraphics{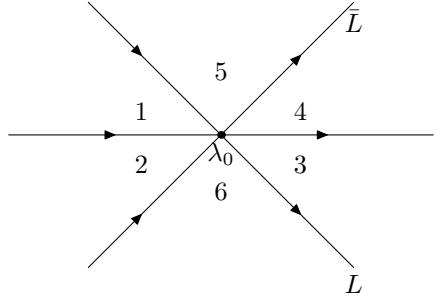}
\caption{The orient contour of $\Sig_1$.}
\end{figure}

Setting
\be
\left\{
\begin{split}
N^{(1')}(\lam)=&N^{(1)}(\lam),\qquad \qquad \lam \in 5\cup 6,\\
=&N^{(1)}(\lam)(b_-^a)_{x,t,\dta}^{-1},\quad \lam \in 2\cup 4,\\
=&N^{(1)}(\lam)(b_+^a)_{x,t,\dta}^{-1},\quad \lam \in 1\cup 3.
\end{split}
\right.
\ee
we find that a simple computation shows that (\ref{N1RHP}) is equivalent to the factorization problem
\be \label{N1'RHP}
\left\{
\ba{cl}
N^{(1')}_+(\lam)=N^{(1')}_-(\lam)(J_{N^{(1')}}(\lam))_{x,t,\dta},&\lam \in \Sig_1,\\
N^{(1')}(\lam)\rightarrow \id,&\lam \rightarrow \infty.
\ea
\right.
\ee
where
\be \label{N1'jump}
\left\{
\ba{rll}
(J_{N^{(1')}}(\lam))_{x,t,\dta}=&(b_-^o)_{x,t,\dta}^{-1}(b_+^o)_{x,t,\dta},&\lam \in \R,\\
=&(b_+^a)_{x,t,\dta},&\lam \in L,\\
=&(b_-^a)_{x,t,\dta}^{-1},&\lam \in \bar L.
\ea
\right.
\ee
Indeed to show that $(b_+^a)_{x,t,\dta}^{-1}$, for example, converges to $\id$ as $\lam \rightarrow \infty$ in $1$ we observe that, for fixed $x,t$, by formula (\ref{h11split}) and the bounded of the function $\dta(\lam)$,we have
\be
\begin{split}
|\dta^2e^{-2it\tha(\lam)}h_2(\lam)|\le & c|\beta(\lam)|e^{-t\re{i\tha(\lam)}}\left|\int_{-\infty}^{t}e^{i(s-t)\tha(\lam)}\widehat{(h/\beta)}(s)\bar ds\right|\\
\le &c\frac{|\lam-\lam_0|^q}{|\lam-i|^{q+2}}\int_{-\infty}^{t}|\widehat{(h/\beta)}(s)|\bar ds\le \frac{C}{|\lam-i|^2}
\end{split}
\ee
and
\be
|\dta^2e^{-2it\tha(\lam)}R(\lam)|\le \frac{C|\sum_{i=0}^{k}\mu_i(\lam-\lam_0)^i|}{|\lam-i|^{k+5}}\le \frac{C}{|\lam-i|^5},
\ee
which is converges to $0$ as $\lam\rightarrow \infty$, and so on.
\par
Set
\be
(w_{\pm}^{1'})_{x,t,\dta}=\pm ((b_{\pm}^{1'})_{x,t,\dta}-\id),
\ee
\be
(w^{1'})_{x,t,\dta}=(w_+^{1'})_{x,t,\dta}+(w_-^{1'})_{x,t,\dta},
\ee
Observe from Proposition \ref{summtec} that, for fixed $x,t$, we then have
\be \label{w1'space}
(w_{\pm}^{1'})_{x,t,\dta},(w^{1'})_{x,t,\dta},\in L^1(\Sig_1)\cap L^{\infty}(\Sig_1).
\ee

\subsection{The second step}
In this section we show how to convert the Riemann-Hilbert problem (\ref{N1'RHP}) on $\Sig_1$ to a Riemann-Hilbert problem on a truncated contour with controlled error terms.
\par
From the above section we have
\be
\begin{split}
\tilde q(x,t)=&2i\lim_{\lam \rightarrow \infty}(\lam N(x,t,\lam))_{12}\\
=&i\lim_{\lam \rightarrow \infty}\lam[\sig_3,N(x,t,\lam)]\\
=&i\lim_{\lam \rightarrow \infty}\lam[\sig_3,N^{(1)}(x,t,\lam)]
\end{split}
\ee
In particular we can take the limit as $\lam \rightarrow\infty$ in $5$, where $N^{(1)}(x,t,\lam)=N^{(1')}(x,t,\lam)$, so
\be \label{qendsolve}
\tilde q(x,t)=i\lim_{\lam \rightarrow \infty}\lam[\sig_3,N^{(1')}(x,t,\lam)].
\ee
\par
The Riemann-Hilbert problem (\ref{N1'RHP}) can be solved as follows (see, for example, \cite{bc}). Let
\be \label{bcsolve1}
(C_{\pm}f)(\lam)=\int_{\Sig_1}\frac{f(\zeta)}{\zeta-\lam_{\pm}}\frac{d\zeta}{2\pi i},\quad \lam\in \Sig_1,f\in L^2(\Sig),
\ee
denote the Cauchy operator on $\Sig_1$ oriented as in Figure 6. Thus, for example, for $\lam>\lam_0$ we have $(C_+f)(\lam)=\lim_{\eps\downarrow 0}\int_{\Sig_1}\frac{f(\zeta)}{\zeta-(\lam-i\eps)}\frac{d\zeta}{2\pi i}$, etc. As is well known, the operators $C_{\pm}$ are bounded from $L^2(\Sig_1)$ to $L^2(\Sig_1)$, and $C_+-C_-=1$. Also, by scaling, we know that the bounds on $C_{\pm}:L^2(\Sig_1)\rightarrow L^2(\Sig_1)$ are independent of $\lam_0$.
\par
Define
\be \label{bcsolve2}
C_{w_{x,t,\dta}^{1'}}f=C_+(f(w_-^{1'})_{x,t,\dta})+C_-(f(w_+^{1'})_{x,t,\dta})
\ee
for $2\times 2$ matrix-valued functions $f$. By property (\ref{w1'space}), $C_{w_{x,t,\dta}^{1'}}$ is a bounded map from $L^2(\Sig_1)+L^{\infty}(\Sig_1)$ into $L^2(\Sig_1)$. Let $\mu^{1'}=\mu^{1'}(\lam;x,t)\in L^2(\Sig_1)+L^{\infty}(\Sig_1)$ be the solution of the basic inverse equation
\be \label{bcbasicequ}
\mu^{1'}=\id+C_{w_{x,t,\dta}^{1'}}\mu^{1'}.
\ee
Then
\be \label{bcsolve3}
N^{(1')}(x,t,\lam)=\id+\int_{\Sig_1}\frac{\mu^{1'}(\zeta;x,t)w^{1'}_{x,t,\dta}(\zeta)}{\zeta-\lam}\frac{d\zeta}{2\pi i},\quad \lam \in \C \backslash \Sig_1,
\ee
is the unique solution of the Riemann-Hilbert problem (\ref{N1'RHP}). Indeed,
\be
\begin{split}
\mu_{\pm}^{1'}=&\id+C_{\pm}(\mu^{1'}w_{x,t,\dta}^{1'})\\
=&\id+C_{\pm}(\mu^{1'}(w^{1'}_+)_{x,t,\dta})+C_{\pm}(\mu^{1'}(w^{1'}_-)_{x,t,\dta})\\
=&\id+C_{w^{1'}_{x,t,\dta}}\mu^{1'}\pm \mu^{1'}(w^{1'}_{\pm})_{x,t,\dta}\\
=&\mu^{1'}(b^{1'}_{\pm})_{x,t,\dta}
\end{split}
\ee
by equation (\ref{bcbasicequ}) and formula (\ref{N1'jump}), which implies that
\[
N^{(1')}_+=N^{(1')}_-(b^{1'}_-)^{-1}_{x,t,\dta}(b^{1'}_+)_{x,t,\dta}=N^{(1')}_-(J_{N^{(1')}})_{x,t,\dta},
\]
as desired. Substituting formula (\ref{bcsolve3}) into (\ref{qendsolve}), we learn that
\be
\begin{split}
q(x,t)=&-(\int_{\Sig_1}[\sig_3,\mu^{1'}(\zeta;x,t)w^{1'}_{x,t,\dta}(\zeta)]\frac{d\zeta}{2\pi i})_{12}\\
=&-(\int_{\Sig_1}[\sig_3,((\id-C_{w^{1'}_{x,t,\dta}})^{-1})(\zeta)w^{1'}_{x,t,\dta}(\zeta)]\frac{d\zeta}{2\pi i})_{12}
\end{split}
\ee
\par
Let $w^e:\Sig_1\rightarrow M(2,\C)$ be a sum of three terms
\be \label{opsum}
w^e=w^a+w^b+w^c.
\ee
we then have the following:
\be
\left\{
\ba{l}
\ba{l}w^a=w^{1'}_{x,t,\dta}|\R \mbox{ is supported on $\R$ and is composed of terms of type}\\ \mbox{\quad $h_1$}.\ea\\
\ba{l}w^b \mbox{ is supported on $L\cup \bar L$ and is composed of the contribution to $w^{1'}_{x,t,\dta}$}\\ \mbox{\quad from terms of type $h_2$.}\ea \\
\ba{l}w^c \mbox{ is supported on $L_{\eps}\cup \bar L_{\eps}$ and is composed of the contribution to $w^{1'}_{x,t,\dta}$}\\ \mbox{\quad from terms of type $R$.}\ea
\ea
\right.
\ee
\begin{figure}[th]
\centering
\includegraphics{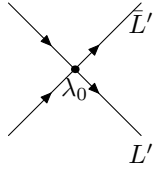}
\caption{The orient contour of cross $\Sig_1'$.}
\end{figure}
Set
\be
\Sig_1^{'}=\Sig_1\backslash (\R \cup L_{\eps}\cup \bar L_{\eps})
\ee
where
\[
L_{\eps}:\{\lam=\lam_0+\lam_0ue^{-i\frac{\pi}{4}};u\ge \eps\}\cup \{\lam=\lam_0+\lam_0ue^{i\frac{3\pi}{4}};u\ge \eps\}
\]

with the orientation as in Figure 7. Define $w'$ through
\be
w^{1'}_{x,t,\dta}=w'+w^e.
\ee
Observe that $w'=0$ on $\Sig_1\backslash \Sig_1^{'}$.
\par
The following estimates are immediate from the bounded of $\dta(\lam)$ and Proposition \ref{summtec}, in which the decay rate $l$ can be chosen to be arbitrarily large. Of course $L^2$ estimates follow immediately from $L^1$ and $L^{\infty}$ estimates. However, throughout this article and in particular in the lemma that follows, we write out the $L^2$ estimates explicitly for the reader's convenience.
\begin{lemma} \label{techlema}
For $\lam_0<M$,
\be
||w^a||_{L^{\infty}(\R)\cap L^2(\R)\cap L^1(\R)}\le Ct^{-l},
\ee
\be
||w^b||_{L^{\infty}(L\cup \bar L)\cap L^2(L\cup \bar L)\cap L^1(L\cup \bar L)}\le Ct^{-l},
\ee
\be
\ba{l}
||w^c||_{L^{\infty}(L_{\eps}\cup \bar L_{\eps})}\le Ce^{-\gam_{\eps}\tau},\\
||w^c||_{L^{2}(L_{\eps}\cup \bar L_{\eps})}\le C\lam_0^{\frac{1}{2}}e^{-\gam_{\eps}\tau},\\
||w^c||_{L^{1}(L_{\eps}\cup \bar L_{\eps})}\le C\lam_0 e^{-\gam_{\eps}\tau},
\ea
\ee
where $\gam_{\eps}=\min{(4\eps^2,\frac{2\eps^2}{3M})}$.
\par
Also,
\be
||w^{1'}_{x,t,\dta}||_{L^2(\Sig_1)},\quad ||w^{e}||_{L^2(\Sig_1)},\quad ||w^{'}||_{L^2(\Sig_1)}\le C.
\ee
\end{lemma}
\begin{proposition}
In the case, $\lam_0<M$ and $\tau\rightarrow \infty$, $(1-C_{w'})^{-1}:L^2(\Sig_1)\rightarrow L^2(\Sig_1)$ exists and is uniformly bounded:
\be \label{prow'bound}
||(1-C_{w'})^{-1}||_{L^2(\Sig_1)}\le C.
\ee
\end{proposition}
\begin{corollary}
In the case, $\lam_0<M$ and $\tau\rightarrow \infty$, $(1-C_{w^{1'}_{x,t,\dta}})^{-1}:L^2(\Sig_1)\rightarrow L^2(\Sig_1)$ exists and is uniformly bounded:
\be
||(1-C_{w^{1'}_{x,t,\dta}})^{-1}||_{L^2(\Sig_1)}\le C.
\ee
\end{corollary}

A simple computation shows that
\be
\ba{rl}
\int_{\Sig_1}((1-C_{w^{1'}_{x,t,\dta}})^{-1}\id)w^{1'}_{x,t,\dta}=&\int_{\Sig_1}((1-C_{w'})^{-1}\id)w'+\int_{\Sig_1}w^e\\
&+\int_{\Sig_1}((1-C_{w'})^{-1}(C_{w'}\id))w^{1'}_{x,t,\dta}\\
&+\int_{\Sig_1}((1-C_{w'})^{-1}(C_{w'}\id))w^e\\
&+\int_{\Sig_1}((1-C_{w'})^{-1}C_{w^e}(1-C_{w^{1'}_{x,t,\dta}})^{-1})\\
&\times (C_{w^{1'}_{x,t,\dta}}\id)w^{1'}_{x,t,\dta}\\
=&\int_{\Sig_1}((1-C_{w'})^{-1}\id)w'+\Rmnum{1}+\Rmnum{2}+\Rmnum{3}+\Rmnum{4}.
\ea
\ee
In the case $\lam_0<M$, from Lemma \ref{techlema} it follows that
\be
\ba{c}
\ba{rl}
|\Rmnum{1}|&\le \|w^a\|_{L^1(\R)}+\|w^b\|_{L^1(L\cup \bar L)}+\|w^c\|_{L^1(L_{\eps}\cup \bar L_{\eps})}\\
&\le Ct^{-l}+Ct^{-l}+C\lam_0\tau^{-l}\\
&\le C\lam_0\tau^{-l},\qquad \qquad \mbox{as $t^l\ge C(\tau^l/\lam_0)$ for $\lam_0<M$},
\ea
\\
\ba{rl}
|\Rmnum{2}|&\le \|(1-C_{w'})^{-1}\|_{L^2(\Sig_1)}\|C_{w^e}\id\|_{L^2(\Sig_1)}\|w^{1'}_{x,t,\dta}\|_{L^2(L_{\Sig_1}}\\
&\le C\|w^e\|_{L^2(\Sig_1)}\|w^{1'}_{x,t,\dta}\|_{L^2(\Sig_1)},\qquad \qquad \mbox{by (\ref{prow'bound})}.
\ea
\ea
\ee
As above, we have
\be \label{webound}
\|w^e\|_{L^2(\Sig_1)}\le ct^{-l}+ct^{-l}+c\sqrt{\lam_0}\tau^{-l}\le c\sqrt{\lam_0}\tau^{-l},
\ee
and
\be
\|w^{1'}_{x,t,\dta}\|_{L^2(\Sig_1)}\le \|w^e\|_{L^2(\Sig_1)}+\|w^{'}\|_{L^2(\Sig_1)}\le c\sqrt{\lam_0}\tau^{-l}+c(t\lam_0)^{-\frac{1}{4}}\le c(t\lam_0)^{-\frac{1}{4}}(1+\tau^{\frac{1}{4}-l}).
\ee
Thus,
\be
|\Rmnum{2}|\le c\lam_0\tau^{-(l+\frac{1}{4})}(1+\tau^{\frac{1}{4}-l}).
\ee
Continuing, we obtain
\be
\ba{rl}
|\Rmnum{3}|&\le \|(1-C_{w'}^{-1})\|_{L^2(\Sig_1)}\|C_{w'}\id\|_{L^2(\Sig_1)}\|w^e\|_{L^2(\Sig_1)}\\
&\le c\|w'\|_{L^2(\Sig_1)}\|w^e\|_{L^2(\Sig_1)}\le c\lam_0\tau^{-(l+\frac{1}{4})},
\ea
\ee
and
\be
\ba{rl}
|\Rmnum{4}|&\le \|(1-C_{w'}^{-1})\|_{L^2(\Sig_1)}\|(1-C_{w^{1'}_{x,t,\dta}}^{-1})\|_{L^2(\Sig_1)}\|C_{w^e}\|_{L^2(\Sig_1)}\|C_{w^{1'}_{x,t,\dta}}\id\|_{L^2(\Sig_1)}\|w^{1'}_{x,t,\dta}\|_{L^2(\Sig_1)}\\
&\le c\|w^e\|_{L^{\infty}(\Sig_1)}\|w^{1'}_{x,t,\dta}\|_{L^2(\Sig_1)}
\ea
\ee
From Lemma \ref{techlema}, for $\lam_0<M$, however, we have $\|w^e\|_{L^{\infty}(\Sig_1)}\le c\tau^{-l}$ and
\be
\ba{rl}
\|w^{1'}_{x,t,\dta}\|^2_{L^2(\Sig_1)}&\le 4(\|w^a\|^2_{L^2(\Sig_1)}+\|w^b\|^2_{L^2(\Sig_1)}+\|w^c\|^2_{L^2(\Sig_1)}+\|w^{'}\|^2_{L^2(\Sig_1)})\\
&\le c(t^{-2l}+t^{-2l}+\lam_0\tau^{-\frac{1}{2}}+(t\lam_0)^{-\frac{1}{2}})\le c(t\lam_0)^{-\frac{1}{2}},
\ea
\ee
which implies that
\be
|\Rmnum{4}|\le c\tau^{-l}(t\lam_0)^{-\frac{1}{2}}
\ee
To summarize, for $\lam_0<M$,
\be
|\Rmnum{1}+\Rmnum{2}+\Rmnum{3}+\Rmnum{4}|\le c\lam_0\tau^{-l}
\ee
as $\tau \rightarrow \infty$.
\par
We have proven the following result:
\begin{lemma}
\be \label{qinfty}
\tilde q(x,t)=(\int_{\Sig_1}[\sig_3,((1-C_{w'})^{-1}\id)(\zeta)w'(\zeta)]\frac{d\zeta}{2\pi i})_{21}+O(\lam_0\tau^{-l}).
\ee
\end{lemma}
\par
We now show that, in general, one may always choose to add or to delete a portion of a contour, on which the jump is $\id$, without altering the Riemann-Hilbert problem in the operator sense(see identities (\ref{ideopra1})-(\ref{ideopra3}) below). Suppose that $\Sig^{(1)}$ and $\Sig^{(2)}$ are two oriented skeletons in $\C$ with
\be
\mbox{card}(\Sig^{(1)}\cap \Sig^{(2)})<\infty;
\ee
let $u=u(\lam)=u_+(\lam)+u_-(\lam)$ be a $2\times 2$ matrix-valued function on
\be
\Sig^{(12)}=\Sig^{(1)}\cup \Sig^{(2)}
\ee
with entries in $L^2(\Sig^{(12)})\cap L^{\infty}(\Sig^{(12)})$ and suppose that
\be
u=0\qquad \qquad \mbox{on }\Sig^{(2)}.
\ee
Let
\be
R_{\Sig^{(1)}}\mbox{ denote the restriction map }L^2(\Sig^{(12)})\rightarrow L^2(\Sig^{(1)}),
\ee
\be
\id_{\Sig^{(1)}\rightarrow \Sig^{(12)}}\mbox{ denote the embedding }L^2(\Sig^{(1)})\rightarrow L^2(\Sig^{(12)}),
\ee
\be
C_{u}^{12}:L^2(\Sig^{(12)})\rightarrow L^2(\Sig^{(12)})\mbox{ denote the operator in (\ref{bcsolve2}) with }u\leftrightarrow w_{x,t,\dta}^{1'},
\ee
\be
C_u^1:L^2(\Sig^{(1)})\rightarrow L^2(\Sig^{(1)})\mbox{ denote the operator in (\ref{bcsolve2}) with }u \uparrow \Sig^{(1)}\leftrightarrow w_{x,t,\dta}^{1'},
\ee
\be
C_u^E:L^2(\Sig^{(1)})\rightarrow L^2(\Sig^{(12)})\mbox{ denote the restriction of $C_u^{12}$ to }L^2(\Sig^{(1)}).
\ee
And, finally, let
\be
\left\{
\ba{l}
\id_{\Sig^{(1)}}\mbox{ and }\id_{\Sig^{(12)}}\mbox{ denote the identity operators on}\\
L^2(\Sig^{(1)})\mbox{ and }L^2(\Sig^{(12)}),\mbox{ respectively}.
\ea
\right.
\ee
We then have the next lemma:
\begin{lemma} \label{opralema}
\be \label{ideopra1}
C_u^{12}C_u^E=C_u^EC_u^{12},
\ee
\be \label{ideopra2}
(\id_{\Sig^{(1)}}-C_u^1)^{-1}=R_{\Sig^{(1)}}(\id_{\Sig^{(12)}}-C_u^{12})^{-1}\id_{\Sig^{(1)}\rightarrow \Sig^{(12)}},
\ee
\be \label{ideopra3}
(\id_{\Sig^{(12)}}-C_u^{12})^{-1}=\id_{\Sig^{(12)}}+C_u^E(\id_{\Sig^{(1)}}-C_u^1)^{-1}R_{\Sig^{(1)}},
\ee
in the sense that if the right-hand side of (\ref{ideopra2}),resp. (\ref{ideopra3}), exists, then the left-hand side exists and identity (\ref{ideopra2}),resp. (\ref{ideopra3}), holds true.
\end{lemma}
\begin{proof}
The proof of identity (\ref{ideopra1}) is trivial. If $g\in L^2(\Sig^{(12)})$ and $\lam \in \Sig^{(1)}$, then
\be
\ba{rl}
((\id_{\Sig^{(1)}}-C_u^1)R_{\Sig^{(1)}}g)(\lam)&=g(\lam)-(C_u^1R_{\Sig^{(1)}}g)(\lam)\\
&=g(\lam)-(C_u^{12}g)(\lam)\\
&=((\id_{\Sig^{(12)}}-C_u^{12})g)(\lam).
\ea
\ee
Hence, for $f\in L^2(\Sig^{(1)})$, we have
\be
\ba{rl}
((\id_{\Sig^{(1)}}-C_u^1)R_{\Sig^{(1)}}(\id_{\Sig^{(12)}}-C_u^{12})^{-1}\id_{\Sig^{(1)}\rightarrow \Sig^{(12)}}f)(\lam)&=((\id_{\Sig^{(12)}}-C_u^{12})\\&((\id_{\Sig^{(12)}}-C_u^{12})^{-1}\id_{\Sig^{(1)}\rightarrow \Sig^{(12)}}f))(\lam)\\
&=(\id_{\Sig^{(1)}\rightarrow \Sig^{(12)}}f)(\lam)=f(\lam).
\ea
\ee
Conversely, if $f\in L^2(\Sig^{(1)})$ and $\lam \in \Sig^{(1)}$, then
\[
(\id_{\Sig^{(1)}\rightarrow \Sig^{(12)}}(\id_{\Sig^{(1)}}-C_u^1)f)(\lam)=((\id_{\Sig^{(1)}\rightarrow \Sig^{(12)}}f)-C_u^{12}(\id_{\Sig^{(1)}\rightarrow \Sig^{(12)}}f))(\lam)
\]
and hence,
\[
(\id_{\Sig^{(1)}\rightarrow \Sig^{(12)}}(\id_{\Sig^{(1)}}-C_u^1)f)(\lam)=(\id_{\Sig^{(12)}}-C_u^{12})(\id_{\Sig^{(1)}\rightarrow \Sig^{(12)}}f)+g_2
\]
where $g_2=0$ on $\Sig^{(1)}$. However $(\id_{\Sig^{(12)}}-C_u^{12})g_2=\id_{\Sig^{(12)}}g_2=g_2$, and so $(\id_{\Sig^{(12)}}-C_u^{12})^{-1}g_2=g_2$. It follows that
\[
\ba{rl}
R_{\Sig^{(1)}}(\id_{\Sig^{(12)}}-C_u^{12})^{-1}\id_{\Sig^{(1)}\rightarrow \Sig^{(12)}}f&=R_{\Sig^{(1)}}(\id_{\Sig^{(12)}}-C_u^{12})^{-1}(((\id_{\Sig^{(12)}}-C_u^{12})\id_{\Sig^{(1)}\rightarrow \Sig^{(12)}}f)+g_2)\\
&=R_{\Sig^{(1)}}\id_{\Sig^{(1)}\rightarrow \Sig^{(12)}}f+R_{\Sig^{(1)}}g_2=f.
\ea
\]
This proves identity (\ref{ideopra2}).
\par
On the other hand, using (\ref{ideopra1}), we get
\[
\ba{lr}
(\id_{\Sig^{(12)}}-C_u^{12})(\id_{\Sig^{(12)}}+C_u^E(\id_{\Sig^{(1)}}-C_u^{1})^{-1}R_{\Sig^{(1)}})&\\
=(\id_{\Sig^{(12)}}-C_u^{12})+C_u^E(\id_{\Sig^{(1)}}-C_u^{1})(\id_{\Sig^{(1)}}-C_u^{1})^{-1}R_{\Sig^{(1)}}&\\
=\id_{\Sig^{(12)}}-C_u^{12}+C_u^ER_{\Sig^{(1)}}=\id_{\Sig^{(12)}}&.
\ea
\]
Conversely, $R_{\Sig^{(1)}}(\id_{\Sig^{(12)}}-C_u^{12})=(\id_{\Sig^{(1)}}-C_u^{1})R_{\Sig^{(1)}}$, and so we have
\[
\ba{lr}
(\id_{\Sig^{(12)}}+C_u^E(\id_{\Sig^{(1)}}-C_u^{1})^{-1}R_{\Sig^{(1)}})(\id_{\Sig^{(12)}}-C_u^{12})&\\
=(\id_{\Sig^{(12)}}-C_u^{12})+C_u^E(\id_{\Sig^{(1)}}-C_u^{1})(\id_{\Sig^{(1)}}-C_u^{1})^{-1}R_{\Sig^{(1)}}&\\
=\id_{\Sig^{(12)}}-C_u^{12}+C_u^ER_{\Sig^{(1)}}=\id_{\Sig^{(12)}}&.
\ea
\]
This proves identity (\ref{ideopra3}) and the lemma.
\end{proof}
We apply Lemma \ref{opralema} to the case $\Sig^{(1)}=\Sig^{'}_1,\Sig^{(12)}=\Sig_1$ and $u=w'$. We learn in particular that
\be
\left\{
\ba{l}
\mbox{the boundedness of $\|(\id-C_{w'})^{-1}\|_{L^2(\Sig_1)}$ is equivalent to}\\
\mbox{the boundedness of $\|(\id^{'}-C'_{w'})^{-1}\|_{L^2(\Sig_1^{'})}$}
\ea
\right.
\ee
(Here $\id'=\id_{\Sig'_1},C'_{w'}=C_{w'}^{\Sig'_1}$.) Also, as in the proof of Lemma \ref{opralema}, from identity (\ref{ideopra2}) we have
\be
\ba{rl}
(\id^{'}-C'_{w'})^{-1}&=R_{\Sig'_1}(\id-C_{w'})^{-1}(\id+g_2),\qquad g_2=0\mbox{ on $\Sig'_1$}\\
&=R_{\Sig'_1}(\id-C_{w'})^{-1}\id+R_{\Sig'_1}g_2\\
&=R_{\Sig'_1}(\id-C_{w'})^{-1}\id.
\ea
\ee
Inserting this identity in formula (\ref{qinfty}) yields the following proposition:
\begin{proposition}
\be \label{qintegral}
\tilde q(x,t)=(\int_{\Sig'_1}[\sig_3,((1'-C'_{w'})^{-1}\id)(\zeta)w'(\zeta)]\frac{d\zeta}{2\pi i})_{21}+O(\lam_0\tau^{-l}).
\ee
\end{proposition}
\par
Set
\be
L'=L\backslash L_{\varplon}.
\ee
Then
\be
\Sig'_1=L'\cup \bar L'.
\ee
On $\Sig'_1$ set $\mu'=(1'-C'_{w'})^{-1}\id$. As in formula (\ref{bcsolve3}), it follows that
\be
N'^{(1')}(\lam)=\id+\int_{\Sig'_1}\frac{\mu'(\zeta)w'(\zeta)}{\zeta-\lam}\frac{d\zeta}{2\pi i}
\ee
solves the Riemann-Hilbert problem
\be
\left\{
\ba{cl}
N'^{(1')}_+(\lam)=N'^{(1')}_-(\lam)(J'_{N^{(1')}})_{x,t,\dta}(\lam),&\lam \in \Sig'_1\\
N'^{(1')}(\lam)\rightarrow \id&\lam\rightarrow \infty.
\ea
\right.
\ee
where
\be
w'=w'_++w'_-,
\ee
\be
b'_{\pm}=\id+w'_{\pm},
\ee
\be
(b'_{\pm})_{x,t,\dta}=\dta^{\hat \sig_3}e^{-it\tha \hat \sig_3}b'_{\pm},
\ee
\be
(J'_{N^{(1')}})_{x,t,\dta}=(b'_-)_{x,t,\dta}^{-1}(b'_+)_{x,t,\dta}.
\ee
From formulae (\ref{bpmdef}) and (\ref{N1'jump}) we can get that
\be
\left\{
\ba{lll}
b_+'=\left\{\ba{ll}\left(\ba{ll}1&(R)^1\\0&1\ea\right),&\lam<\lam_0\\
\left(\ba{ll}1&(R)^3\\0&1\ea\right),&\lam>\lam_0\ea\right.,
&b_-'=\left(\ba{ll}1&0\\0&1\ea\right),&\mbox{on } L'\\
b_+'=\left(\ba{ll}1&0\\0&1\ea\right),&b_-'=\left\{\ba{ll}\left(\ba{ll}1&(R)^2\\0&1\ea\right),&\lam<\lam_0\\
\left(\ba{ll}1&(R)^4\\0&1\ea\right),&\lam>\lam_0\ea\right.,
&\mbox{on } \bar L'
\ea
\right.
\ee
\subsection{The next step}
Extend $\Sig'_1$ to the contour
\be
\hat \Sig'_1=\{\lam=\lam_0+\lam_0ue^{\pm i\frac{\pi}{4}},u\in \R\}\cup \{\lam=\lam_0+\lam_0ue^{\pm i\frac{3\pi}{4}},u\in \R\}.
\ee

\par
Define the scaling operator
\[
N:L^2(\Sig_1)\rightarrow L^2(\Sig_1\backslash \lam_0)
\]
\be
f(\lam)\mapsto Nf(\lam)=f(\frac{\lam}{\sqrt{8t}}+\lam_0)
\ee
\par
Then we can get
\be
\tilde q(x,t)=\frac{1}{\sqrt{2\pi}}(8t)^{-i\nu}e^{4it\lam_0^2+2i\gam}(n_1^0)_{12}+o(t^{-\frac{1}{2}}).
\ee
where $\gam=\frac{1}{2\pi}\int_{-\infty}^{\lam_0}\log |\lam-\lam_0|d\log(1-\lam|\rho(\lam)|^2)$.
\par
We notice that in the neighborhood of the stationary point $\lam=\lam_0$, the function $\dta(\lam)$ appearing in the formula (\ref{dtasol}) can be represented as(\cite{diz})
\be \label{dtasolasyp}
\dta(\lam)_{\pm}=(\lam-\lam_0)^{i\nu}_{\pm}e^{\frac{i}{2\pi}\int_{-\infty}^{\lam_0}\log{|\lam'-\lam_0|d\log{(1-\lam'|\rho(\lam')|^2)}}}
\ee
where
\be
\nu=-\frac{1}{2\pi}\log(1-\lam_0|\rho(\lam_0)|^2)
\ee
and $(\lam-\lam_0)_{\pm}^{i\nu}$ denotes the boundary values of the corresponding multivalued function defined on the $\lam-$plane with the cut along $(-\infty,\lam]$.
\par
Then a straightforward computation shows that as $t\rightarrow \infty$
\be \label{asymptotic}
(N\dta^{\hat \sig_3} e^{-it\tha \hat \sig_3}[f])(\lam)\rightarrow \phi^{\hat \sig_3}\lam^{\nu i \hat \sig_3}e^{-i\frac{\lam^2}{4}\hat \sig_3}[f](\lam_0),
\ee
where
\be
\phi=(8t)^{-\frac{i\nu}{2}}e^{2it\lam_0^2}e^{\frac{i}{2\pi}\int_{-\infty}^{\lam_0}\log(\lam_0-\zeta)d\log(1-\zeta|\rho(\zeta)|^2)}
\ee
and
$[f](\lam_0)$ is defined by
$[f]=R+h_1$
where $R$ and $h_1$ are defined in (\ref{case}) in subsection \ref{first}.

\par
It follows from the exponential decay of $e^{-\frac{i\lam^2}{4}\hat \sig_3}[f](\lam_0)$, that the asymptotic formula in (\ref{asymptotic}) has an $L^1\cap L^2\cap L^{\infty}(\Sig_1-\lam_0)$ error of order $\frac{\log t}{t^{\frac{1}{2}}}$. Since $\phi$ is indenpendent of $\lam$, $N^0$ is the solution of the Riemann-Hilbert problem on $\Sig_1-\lam_0$,
\be \label{N0RHP}
\left\{
\ba{l}
N^0_+=N^0_-\lam^{i\nu \hat \sig_3}e^{-\frac{i\lam^2}{4}\hat \sig_3}[f](\lam_0),\\
N^0\rightarrow \id,\qquad \quad \mbox{as }\lam\rightarrow \infty.
\ea
\right.
\ee
if and only if $\phi^{\hat \sig_3}N^0$ is the solution of the Riemann-Hilbert problem for the jump matrix given by the right-hand side of (\ref{asymptotic}). Deforming the Riemann-Hilbert problem (\ref{N0RHP}) on $\Sig_1-\lam_0$ to the real axis we obtain the Riemann-Hilbert problem
\be \label{modelRHP}
\left\{
\ba{l}
N^0_+=N^0_-e^{-\frac{i\lam^2}{4}\hat \sig_3}\lam_-^{i\nu \sig_3}J(\lam_0)\lam_+^{-i\nu \sig_3},\quad \lam \in \R,\\
N^0\rightarrow \id \qquad \quad\mbox{as }\lam\rightarrow \infty.
\ea
\right.
\ee
which can be solved in closed form. This problem was first considered by Its, and the following calculations can be found in \cite{i} \cite{diz} \cite{dz}.
\par
Setting
\be
\psi(\lam)=N^0(\lam)\lam^{i\nu \sig_3}e^{-\frac{i\lam^2}{4}\sig_3},
\ee
we can represent the Riemann-Hilbert problem (\ref{modelRHP}) as
\be \label{modelRHP2}
\ba{l}
\psi_+(\lam)=\psi_-(\lam)[f](\lam_0),\qquad \lam\in\R,\\
\psi(\lam)\rightarrow \lam^{i\nu \sig_3}e^{-\frac{i\lam^2}{4}\sig_3},\quad \mbox{as }\lam\rightarrow \infty.
\ea
\ee
By differentiation we have
\be
(\frac{d\psi}{d\lam}+\frac{1}{2}i\lam\sig_3\psi)_+=(\frac{d\psi}{d\lam}+\frac{1}{2}i\lam\sig_3\psi)_-[f](\lam_0),\quad \lam\in\R,
\ee
Now as $\det{[f](\lam_0)}=1$, it follows by a Liouville argument that $\det{\psi}=1$. Hence, $\psi^{-1}$ exists and is bounded. But $(\frac{d\psi}{d\lam}+\frac{1}{2}i\lam\sig_3\psi)\psi^{-1}$ has no jump across $\R$ and must be entire. Also
\be
\ba{rl}
(\frac{d\psi}{d\lam}+\frac{1}{2}i\lam\sig_3\psi)\psi^{-1}=&(\frac{dN^0}{d\lam})\psi^{-1}+N^0(i\nu \sig_3\lam^{-1})(N^{0})^{-1}\\
&+N^0(-\frac{i}{2}\lam)\sig_3(N^0)^{-1}+(\frac{i\lam}{2}\sig_3N^0)(N^0)^{-1}\\
=&O(\lam^{-1})+\frac{1}{2}i\lam[\sig_3,N^0](N^0)^{-1}\\
=&O(\lam^{-1})+\frac{1}{2}i[\sig_3,N_1^0]
\ea
\ee
It follows by Liouville's argument that
\be \label{modelbaseq}
\frac{d\psi}{d\lam}+\frac{1}{2}i\lam\sig_3\psi=\beta \psi,
\ee
where
\be
\beta=\frac{i}{2}[\sig_3,N_1^0]=\left(\ba{cc}0&\beta_{12}\\\beta_{21}&0\ea\right).
\ee
In particular,
\be
(N_1^0)_{12}=-i\beta_{12}.
\ee
\par
Consider first $\im{\lam}>0$. From equation (\ref{modelbaseq}) we obtain
\be
\frac{d^2\psi_{11}^+}{d\lam^2}=(-\frac{\lam^2}{4}-\frac{i}{2}+\beta_{12}\beta_{21})\psi_{11}^+.
\ee
Setting
\be \label{modelbaseqdiff}
\psi_{11}^+(\lam)=g(e^{-\frac{3i\pi}{4}}\lam)
\ee
results in the reduction of equation (\ref{modelbaseqdiff}) to the Weber's equation \cite{whwa}.
\be
\frac{d^2g}{d\zeta^2}+(\frac{1}{2}-\frac{\zeta^2}{4}+a)g(\zeta)=0,
\ee
where $a=i\beta_{12}\beta_{21}$.
\par
Hence,
\be \label{psiweb}
\psi_{11}^+(\lam)=c_1D_a(e^{-\frac{3i\pi}{4}}\lam)+c_2D_a(-e^{-\frac{3i\pi}{4}}\lam)
\ee
where $D_a(\cdot)$ denotes the denotes the standard (entire) parabolic-cylinder function.
\par
From \cite{whwa}, we know that as $\zeta\rightarrow \infty$,
\be
\ba{l}
D_a(\zeta)=\zeta^ae^{-\frac{1}{4}\zeta^2}(1+O(\zeta^{-2})),\qquad |\arg \zeta|<\frac {3\pi}{4},\\
=\zeta^ae^{-\frac{1}{4}\zeta^2}(1+O(\zeta^{-2}))-\frac{\sqrt{2\pi}}{\Gam(-a)}e^{a\pi i}\zeta^{-a-1}e^{\frac{1}{4}\zeta^2}(1+O(\zeta^{-2})),\qquad \frac{\pi}{4}<\arg \zeta<\frac{5\pi}{4},\\
=\zeta^ae^{-\frac{1}{4}\zeta^2}(1+O(\zeta^{-2}))-\frac{\sqrt{2\pi}}{\Gam(-a)}e^{-a\pi i}\zeta^{-a-1}e^{\frac{1}{4}\zeta^2}(1+O(\zeta^{-2})),\qquad -\frac{5\pi}{4}<\arg \zeta<-\frac{\pi}{4},
\ea
\ee
where $\Gam$ is the Gamma function.
\par
Setting $\lam=e^{\frac{3\pi i}{4}}\sig$, with $\sig>0$, comparing the right-hand side and the left-hand side of (\ref{psiweb}), we conclude that
\be
c_2=0,\quad a=i\nu \quad \mbox{and }c_1=e^{-\frac{3\pi \nu}{4}},
\ee
so that
\be
\psi_{11}^+(\lam)=e^{-\frac{3\pi \nu}{4}}D_a(e^{-\frac{3i\pi}{4}}\lam),\qquad \im \lam>0.
\ee
From equation (\ref{modelbaseq}) we learn that
\be
\ba{ll}
\psi_{21}^+(\lam)&=\frac{1}{\beta_{12}}(\frac{d\psi_{11}}{d\lam}+\frac{1}{2}i\lam\psi_{11})\\
&=\frac{1}{\beta_{12}}e^{-\frac{3\pi \nu}{4}}(\frac{d}{d\lam}(D_a(e^{-\frac{3i\pi}{4}}\lam))+\frac{1}{2}i\lam D_a(e^{-\frac{3i\pi}{4}}\lam)),\im \lam>0.
\ea
\ee
\par
In a similar way we find that
\be
\psi_{22}^+(\lam)=e^{\frac{\pi \nu}{4}}D_{-a}(e^{-\frac{i\pi}{4}}\lam),\qquad \im \lam>0,
\ee
and
\be
\psi_{12}^+(\lam)=\frac{1}{\beta_{21}}e^{\frac{\pi \nu}{4}}(\frac{d}{d\lam}(D_{-a}(e^{-\frac{i\pi}{4}}\lam))-\frac{1}{2}i\lam D_{-a}(e^{-\frac{i\pi}{4}}\lam)),\im \lam>0.
\ee
\par
Repeating these calculations for $\im \lam<0$ yields
\bea
&&\psi_{11}^-(\lam)=e^{\frac{\pi \nu}{4}}D_a(e^{\frac{i\pi}{4}}\lam),\\
&&\psi_{21}^-(\lam)=\frac{1}{\beta_{12}}e^{\frac{\pi \nu}{4}}(\frac{d}{d\lam}(D_a(e^{\frac{i\pi}{4}}\lam))+\frac{1}{2}i\lam D_a(e^{\frac{i\pi}{4}}\lam)),\\
&&\psi_{22}^-(\lam)=e^{-\frac{3\pi \nu}{4}}D_{-a}(e^{\frac{3i\pi}{4}}\lam),\\
&&\psi_{12}^-(\lam)=\frac{1}{\beta_{21}}e^{-\frac{3\pi \nu}{4}}(\frac{d}{d\lam}(D_{-a}(e^{\frac{3i\pi}{4}}\lam))-\frac{1}{2}i\lam D_{-a}(e^{\frac{3i\pi}{4}}\lam)).
\eea
Substituting these formulae in (\ref{modelRHP2}) gives
\be
(\psi_-)^{-1}\psi_+=f(\lam_0)=\left(\ba{cc}1-\lam_0|\rho(\lam_0)|^2&-\ol{\rho(\lam_0)}\\\lam_0\rho(\lam_0)&1\ea \right)
\ee
then we can obtain
\be
\ba{rl}
\lam_0\rho(\lam_0)=&\psi_{11}^{-}\psi_{21}^+-\psi_{21}^-\psi_{11}^+\\
=&(e^{\frac{\pi \nu}{4}}D_a(e^{\frac{i\pi}{4}}\lam))(\frac{1}{\beta_{12}}e^{-\frac{3\pi \nu}{4}}(\frac{d}{d\lam}(D_a(e^{-\frac{3i\pi}{4}}\lam))+\frac{1}{2}i\lam D_a(e^{-\frac{3i\pi}{4}}\lam)))\\
&-\frac{1}{\beta_{12}}e^{\frac{\pi \nu}{4}}(\frac{d}{d\lam}(D_a(e^{\frac{i\pi}{4}}\lam))+\frac{1}{2}i\lam D_a(e^{\frac{i\pi}{4}}\lam))e^{-\frac{3\pi \nu}{4}}D_a(e^{-\frac{3i\pi}{4}}\lam)\\
=&\frac{e^{-\frac{\pi \nu}{2}}}{\beta_{12}}W(D_a(e^{\frac{i\pi}{4}}\lam)),D_a(e^{-\frac{3i\pi}{4}}\lam))\\
=&\frac{\sqrt{2\pi}e^{i\frac{\pi}{4}}e^{-\frac{\pi \nu}{2}}}{\beta_{12}\Gam(-a)}.
\ea
\ee
where $W(f,g)=fg'-f'g$ is the Wronskian of $f$ and $g$.
\par
Thus
\be
\beta_{12}=\frac{\sqrt{2\pi}e^{i\frac{\pi}{4}}e^{-\frac{\pi \nu}{2}}}{\lam_0\rho(\lam_0)\Gam(-a)}
\ee

\par
Finally, careful bookkeeping of the error terms and the Lemma 8.1 in \cite{avkahv} yields the following result.
\begin{theorem}
Let $q(x,t)$ be the solution of the Cauchy problem (\ref{DNLSandIv}). Then as $t\rightarrow \infty$
\be
q(x,t)=q_{as}(x,t)+O(\frac{\log t}{t})
\ee
where
\be
\ba{l}
q_{as}=\frac{1}{\sqrt{t}}\alpha(\lam_0)e^{\frac{ix^2}{4t}-i\nu(\lam_0)\log t},\\
|\alpha(\lam_0)|^2=\frac{\nu(\lam_0)}{2}=-\frac{1}{4\pi}\log(1-\lam_0|\rho(\lam_0)|^2),\\
\arg \alpha(\lam_0)=-3\nu\log 2-\frac{\pi}{4}+\arg \Gam(i\nu)-\arg r(\lam_0)+\frac{1}{\pi}\int_{-\infty}^{\lam_0}\log|\lam-\lam_0|d\log(1-\lam|\rho(\lam)|^2),\\
\lam_0=-\frac{x}{4t}.
\ea
\ee
\end{theorem}


\begin{thebibliography}{XXXX}

\bibitem{m} E. Mjolhus,{\em On the modulational instability of hydromagnetic waves parallel to the magnetic field}, J. Plasma Phys. {\bf 16}(1976), 321-334.

\bibitem{k} Y. Kodama, {\em Optical solitons in a monomode fiber}, J. Stat. Phys. {\bf 39} (1985), 597-614.

\bibitem{a} G. P. Agrawal, {\em Nonlinear Fiber Optics}, Academic Press, 2007.

\bibitem{l} J. Lenells, {\em The derivative nonlinear Schr\"odinger equation on the half-line}, Physica D {\bf 237}(2008), 3008--3019.

\bibitem{bc} R. Beals and R. Coifman, {\em Scattering and inverse scattering for first order systems},  Comm. in Pure and Applied Math. {\bf 37}(1984), 39--90.

\bibitem{dz} P. Deift and  X. Zhou, {\em A steepest descent method for oscillatory Riemann--Hilbert problems},  Ann. of Math. (2) {\bf 137}(1993), 295-368.

\bibitem{i} A. R. Its, {\em Asymptotics of solutions of the nonlinear Schr\"odinger equation and isomonodromic deformations of systems of linear differential equations }, Sov. Math. Dokl. {\bf 24} (1981), 452-456.

\bibitem{diz} P. A. Deift, A. R. Its, and X. Zhou, {\em Long-time asymptotics for integrable nonlinear wave equations},  in ``Important developments in soliton theory'', 181-204, Springer Ser. Nonlinear Dynam., Springer, Berlin, 1993.

\bibitem{kn} D. J. Kaup and A. C. Newell, {\em An exact solution for a derivative nonlinear Schr\"odinger equation}, J. Math. Phys. {\bf 19}(1978), 789-801.

\bibitem{ki} T. Kawata and H. Inoue, {\em Exact solutions of the derivative nonlinear Schr\"odinger equation under the nonvanishing conditions}, J. Phys. Soc. Japan {\bf 44}(1978), 1968-1976.

\bibitem{ikws} Y. H. Ichikawa, K. Konno, M. Wadati and H. Sanuki, {\em Spiky soliton in circular polarized Alfve wave}, J. Phys. Soc. Japan {\bf 48}(1980), 279-286.

\bibitem{vml} V. M. Lashkin, {\em N-soliton solutions and perturbation theory for the derivative nonlinear Schr\"odinger equation with
           nonvanishing boundary conditions}, J. Phys. A, {\bf 40}(2007), 6119-6132.

\bibitem{mz}  W. X. Ma and R. G. Zhou, {\em On inverse recursion operator and tri-Hamiltonian formulation for a Kaup-Newell system of DNLS equations}, J. Phys. A {\bf 32}(1999), L239-L242.

\bibitem{xhw} S. W. Xu, Jingsong He and Lihong Wang, {\em The Darboux transformation of the derivative nonlinear Schr\"odinger equation} J. Phys. A {\bf 44}(2011), 305203-305225.

\bibitem{fan} E. G. Fan, {\em Darboux transformation and soliton-like solutions for the Gerdjikov-Ivanov equation}, J. Phys. A, {\bf 33}(2000), 6925-6933.

\bibitem{zs} V. E. Zakharov and A. Shabat, {\em A scheme for integrating the nonlinear equations of mathematical physics by the method of the inverse scattering problem,I and II}, Funct. Anal. Appl. {\bf 8}(1974), 226-235 and {\bf 13}(1979), 166-174.

\bibitem{whwa} E.T. WHITTAKER and G.N. WATSON, {\em A Course of Modern Analysis}, 4th ed., Cambridge University Press, Cambridge, 1927.

\bibitem{pdl} P. D. Lax, {\em Integrals of nonlinear equations of evolution and solitary waves}, Comm. Pure. Appl. Math.{\bf 21}(1968), 467-490.

\bibitem{avkahv} A. V. Kitaev and A. H. Vartanian, {\em Leading-order temporal asymptotics of the modified nonlinear Schr\"odinger equation:solitonless sector.}, Inverse Problems {\bf 13}(1997),1311-1339.

\bibitem{chen2}  H. H. Chen, Y. C. Lee and C. S. Liu, {\em Integrability of nonlinear Hamiltonian systems by inverse scattering method}, Phys. Scr. 20(1979), 490-492.

\bibitem{kundu}  A. Kundu, W. Strampp and W. Oevel, {\em Gauge transformations of constrained KP flows: new integrable hierarchies}, J. Math. Phys. 36(1995), 2972-2984.

\bibitem{fan2} E. G. Fan, {\em A family of completely integrable multi-Hamiltonian systems explicitly related to some celebrated equations}, J. Math. Phys. 42(2001), 4327-4344.

\bibitem{xf} J. Xu and E. G. Fan, {\em The derivative nonlinear Schr\"odinger equation on the interval}, arXiv:1205.1559.

\end{thebibliography}
\end{document}